\documentclass[reqno,11pt]{amsart}
\usepackage{geometry}

\usepackage{amssymb,amsfonts,amsthm}
\usepackage[english]{babel}
\usepackage[pdftex]{graphicx} 
\usepackage[square,comma,numbers]{natbib}
\usepackage{etoolbox}
\usepackage{enumerate}
\usepackage{amsmath,latexsym}
\usepackage{amsmath}
\usepackage{amsfonts}
\usepackage{amssymb}
\usepackage{hyperref}
\usepackage{enumitem}
\usepackage{color}
\usepackage{mathrsfs}
\usepackage{soul}
\usepackage[colorinlistoftodos,textsize=footnotesize,color=blue!25]{todonotes}

\setstcolor{red}
\newtheorem{prop}{Proposition}
\newtheorem{defi}[prop]{Definition}

\newtheorem{rem}[prop]{Remark}
\newtheorem{lemma}[prop]{Lemma}
\newtheorem{thm}[prop]{Theorem}

\usepackage{bbm}

\newcommand{\no}[1]{\mathop{\mathopen: {#1} \mathclose:}}

\def\be             {\begin{equation}}
	\def\ee             {\end{equation}}

\newcommand{\beqa}{\begin{eqnarray*}}
	\newcommand{\eeqa}{\end{eqnarray*}}

\newcommand{\beqan}{\begin{eqnarray}}
	\newcommand{\eeqan}{\end{eqnarray}}

\newcommand{\x}{\mathbf{x}}

\long\def\labl#1   {\label{#1}\ee}

\newcommand{\beq}{\begin{equation}}
\newcommand{\eeq}{\end{equation}}

\newcommand{\benu}{\begin{enumerate}}
	\newcommand{\eenu}{\end{enumerate}}

\newcommand{\sst}[1]{\scriptscriptstyle{#1}}  

\newcommand{\C}{\mathbb C}
\newcommand{\R}{\mathbb R}

\newcommand{\M}{\mathbb M}

\newcommand{\loc}{{\rm loc}}

\newcommand{\const}{{\rm const}}

\DeclareMathOperator{\WF}{\rm WF}
\DeclareMathOperator{\supp}{\rm supp}

\newcommand{\fA}{\mf A}

\newcommand{\ph}{\phi}

\newcommand{\tsf}[2]{{\textstyle{\frac{#1}{#2}}}}

\newcommand{\mf}[1]{\mathfrak{#1}}
\newcommand{\mc}[1]{\mathcal{#1}}

\newcommand{\phic}{\varphi}

\newcommand{\cD}{\mc D}
\newcommand{\cDd}{\cD_{\rm den}}

\newcommand{\cE}{\mc E}
\newcommand{\cF}{\mc F}

\newcommand{\cH}{\mc H}
\newcommand{\cT}{\mc T}

\newcommand{\cS}{\mc S}
\newcommand{\Ci}{{\mc C}^\infty}
\newcommand{\cO}{\mc O}

\newcommand{\sD}{\mathscr{D}}
\newcommand{\Ppsi}{\Phi_{\psi}}
\newcommand{\T}{\cdot_{{}^\cT}}

\newcommand{\Psifxi}{\pi_{\psi}(e^{i\Phi(f)})\Omega_0\otimes \xi}

\title{Local nets of von Neumann algebras in the Sine-Gordon model}
\author{Dorothea Bahns, Klaus Fredenhagen, Kasia Rejzner}

\begin{document}
\maketitle
\begin{abstract}
The Haag-Kastler net of local von Neumann algebras is constructed in the ultraviolet finite regime of the sine-Gordon model, and its equivalence with the massive Thirring model is proved.
In contrast to 
other authors, we do not add an auxiliary mass term, and we work completely in Lorentzian signature. The construction is based on the functional formalism for perturbative Algebraic Quantum Field Theory together with estimates originally derived within Constructive Quantum Field Theory and adapted to Lorentzian signature. The paper extends previous work by two of us.
\end{abstract}
\tableofcontents

	
\section{Introduction}

The classical sine-Gordon model is one of the most interesting integrable field theories, and its quantization has been treated since long by many authors
(see e.g. \cite{Faddeev}).
 One strategy is the ansatz with factorizing S-matrices, where the integrable structure is exploited. The corresponding local fields are approached in the so-called form factor program, which, however, has problems in proving the convergence of the arising series \cite{Karowski, Smirnov}. A direct construction of the local von Neumann algebras has been successfully carried out for similar models by Lechner \cite{Lechner}, based on ideas of Schroer \cite{Schroer} and Buchholz, in analogy with \cite{Brunetti}. For models closer to the sine-Gordon model, where this construction is not directly applicable, encouraging results have been found by Cadamuro and Tanimoto \cite{Cadamuro}.

Another strategy is the construction of the model by methods of Constructive Quantum Field Theory. This was performed by Fr\"ohlich and Seiler \cite{Frohlich-Seiler, Frohlich} within the framework of Euclidean Quantum Field Theory. But their methods required the introduction of an extra mass term, or alternatively, a spatial cutoff, due to the infrared problems of the massless free scalar field. It turned out to be difficult to remove the mass term at the end, and also the integrable structure was not visible  in their construction.  In the more recent paper by Benfatto et al. \cite{Benfatto} the equivalence to the Thirring model was shown for the Euclidean theory using a finite volume cutoff.

In an earlier paper  \cite{BR16}, two of us showed that the perturbative expansion of the $S$-matrix with a spacetime cutoff, as well as that of the corresponding  interacting fields, converge. This was achieved in the framework of perturbative Algebraic Quantum Field Theory (pAQFT) \cite{Rejzner}, i.e. on the level of functionals on the theory's configuration space, without a particular choice of a state (viz. a representation on a Hilbert space).  In this paper, we extend these results in the following way. Since the vacuum state of the massless free scalar field is not a regular state on the Weyl algebra of the field, we use a representation introduced by Derezinski and Meissner \cite{DM06}, quite similar to the representation used in early day string theory. We show that in this representation the S-matrix (as a generating functional for time ordered products of the interaction Lagrangian) is unitary and satisfies Bogoliubov's causal factorization condition. Based on this property we construct a family of unitary operators - the relative S-matrices -  which generate the local algebras of observables of the model.

We then discuss the equivalence with the massive Thirring model, first observed by Coleman \cite{Coleman75}. In the functional formalism, we give  an explicit construction of the massless Thirring model within the theory of the massless free scalar field. The equivalence of the massive case with the sine-Gordon model then becomes evident since the interaction Lagrangians coincide.

The paper is organized as follows: We first review the Derezinski-Meissner representation of the free massless scalar field and prove that, as a representation of the canonical commutation relations of time zero fields and their conjugate momenta, it is locally quasiequivalent to the vacuum representations of the massive free scalar fields. The local quasiequivalence between massive theories of different masses was shown long ago by Eckmann and Fr\"ohlich \cite{EF74}. The expectation is that the sine-Gordon theory is massive, and the result on its local quasiequivalence with the massive free theory suggests that indeed, the local von Neumann algebras generated by the relative S-matrices in the Derezinski-Meissner representation coincide with those which one would obtain in a vacuum representation of the model that however still needs to be constructed.  

In the following section we review and extend the construction of local S-matrices from \cite{BR16} and prove that they are unitary operators satisfying Bogoliubov's causal factorization  relation.
We construct bounded operators of the interacting theory corresponding to time ordered exponentials of the field and of vertex operators. 

In the last section we discuss the relation to the Thirring model by a rigorous version of a construction originally described by Mandelstam \cite{Mandelstam} (see also \cite{Hadjiivanov} for a detailed description) . For this purpose, we extend the theory of the free massless scalar field $\Phi$ by adding a dual field $\tilde{\Phi}$ with $\partial_{\mu}\Phi=-\epsilon_{\mu\nu}\partial^{\nu}\tilde{\Phi}$. In this extended theory fermionic fields can be defined which satisfy the field equation of the massless Thirring model. Moreover, the interaction Lagrangian of the sine-Gordon model is shown to agree with the fermionic mass term.
Therefore, the convergence result of the previous section immediately implies the convergence of the mass expansion of the Thirring model,  both in the representation induced by the Derezinski-Meissner representation, and in the vacuum representation of the massless Thirring model used in Coleman's original argument.                

\section{Free massless scalar field in 2 dimensions}
The free massless field in 2 dimensions is probably the simplest field theory one can think of. The equation of motion 
\be\square \phi=0\ee
has the general solution
\be\phi(t,\x)=\phi_L(t+\x)+\phi_R(t-\x)\ee
with arbitrary functions $\phi_L$ and $\phi_R$.

Surprisingly, the corresponding quantum field theory has some features which do not fit into the standard formalism of quantum field theory. There, quantization of a free field theory usually starts by interpreting the solutions 
with positive frequency as wave functions of particles. The Fourier transform of these wave functions are supported on the positive mass shell in momentum space. By using the (up to normalization) unique Lorentz invariant measure on the mass shell one equips the space of wave functions with a positive definite scalar product and obtains the Hilbert space of single particle states. The quantum field can then be defined in terms of annihilation and creation operators on the Fock space over the single particle space.

In 2 dimensions, however, the only Lorentz invariant  measure on the positive light cone (the mass shell for zero mass) is the Dirac measure concentrated at $p=0$.  
As a consequence, the standard construction breaks down, which often is interpreted as saying that the massless scalar field in 2 dimensions does not exist. Indeed, the Wightman axioms cannot be satisfied, and also the Osterwalder-Schrader axioms for the euclidean version cannot be fulfilled.

The problem disappears when one looks instead at the derivatives of the field. For them the standard quantization on a Fock space exists. In order to include also the field itself it is useful to adopt the algebraic point of view of first constructing the algebra of quantum fields and to investigate in a second step the states (defined as normalized positive functionals on the algebra) or, equivalently, the representations on Hilbert spaces. As a result, one finds, that the algebra of the free field exists, but does not possess a vacuum state (whose existence is one of the Wightman axioms). So the problems mentioned above find their natural explanation in the fact that the state space does not have the properties expected from our experience with other free scalar fields (massive or in higher dimensions).

These facts seem to be well known among experts. There is, however, a lack of explicit examples of states for the free field satisfying some natural requirements. From our experience with field theories on generic Lorentzian spacetimes a good class of states are quasifree states where the 2-point function satisfies the so-called Hadamard condition. We are aware of only two places where such states were explicitly constructed: one is the diploma thesis of Sebastian Schubert \cite{Schubert}, where he found a surprisingly simple example of a Hadamard 2-point function. The other place is a paper of Derezinski and Meissner \cite{DM06}, where they explicitly construct a representation on a separable Hilbert space. Actually, Schubert's states are vector states in this representation.

\subsection{The algebra}\label{Algebra}
As customary in pAQFT, we start from the space of smooth field configurations $\mathcal{E}(\M)=\Ci(\M,\R)$ on 2 dimensional Minkowski space $\M$. To fix the notation, let $\mathcal{D}(\M)=\Ci_c(\M,\R)$ be the space of compactly supported functions and $\cDd(\M)$, the space of compactly supported densities (test densities).

 The field equation induces a Poisson structure on the configuration space, given by the difference $\Delta$ of the advanced and the retarded propagator of the operator $P=-\Box$ (i.e. minus the d'Alembertian) on $\mathcal{E}(\M)$ and
\be\label{eq:Delta}
\Delta(x,y)=-\tsf 1 2 \Theta(x^0-y^0-|x^1-y^1|)+\tsf 1 2 \Theta(-x^0+y^0-|x^1-y^1|)\,.
\ee
$\Delta$ induces a linear map (also denoted by $\Delta$) from $\cDd(\M)$ to $\mathcal{E}_{\rm Sol}(\M)\subset\cE(\M)$ the space of
smooth solutions of the equation $P\ph=0$, by the prescription
\[
\Delta f(x)\doteq\int_y \Delta(x,y)f(y)\,,\quad f\in\cD_{\rm den}(\M)\,.
\]
The Poisson bracket of two smooth functions $F,G$ on $\mathcal{E}(\M)$, which have compact spacetime support and 
smooth first functional derivatives $\frac{\delta F}{\delta\ph},\frac{\delta G}{\delta\ph}\in\cD_{\rm den}(\M)$, is then 
\[\{F,G\}=\text{\small $\int$}\frac{\delta F}{\delta\ph}\left(\Delta\frac{\delta G}{\delta\ph}\right)\equiv \left<F^{(1)},\Delta G^{(1)}\right>\,,\]
where the notation $\left<.,.\right>$ emphasizes the duality between $\cDd$ and $\cE$.

The algebra of the quantum field is defined by deformation quantization, by looking for a family of associative products $\star_{\hbar}$ of functionals on $\mathcal{E}(\M)$ such that in the limit $\hbar\to 0$ 
\be F\star_{\hbar}G\to FG\ee
and
\be\frac{1}{i\hbar}(F\star_{\hbar}G-G\star_{\hbar}F)\to\{F,G\}\ .\ee

An example for such a $\star_\hbar$-product is the Weyl-Moyal product
\be (F\star G)[\ph]=
\sum_{n=0}^\infty \frac{\hbar^n}{n!}\left<F^{(n)}[\ph],(\tfrac i 2 \Delta)^{\otimes n} G^{(n)}[\ph]\right>\ .\ee
defined as a formal power series in $\hbar$ for {\em regular} functionals $F, G \in \cF_{\rm reg}(\M)$, that is, the space of functionals on $\mc E(\M)$ whose functional derivatives exist and are compactly supported smooth densities, $F^{(n)}[\ph],G^{(n)}[\ph]\in\cD_{\rm den}(\M^n)$.

Other star products are obtained by adding a symmetric bisolution of $P$ to $\frac{i}{2}\Delta$. Poincar\'{e} invariant symmetric bisolutions are multiples of 
\be\label{Hmu}
H_\mu(x,y)=-\frac{1}{4\pi}\ln(\mu^2|(x-y)^2|)
\ee
with a mass scale $\mu$ which is a dimensionful quantity which takes positive values measured in inverse length.

It turns out that the sums
\[
W_{\mu}=\frac{i}{2}\Delta + H_\mu=-\frac 1 {4\pi} \ln\left(\mu^2(-x\cdot x) + i \mu \varepsilon t\right)
\]
satisfy the so-called Hadamard condition, i.e. their wavefront sets fulfill the microlocal spectrum condition \cite{Radzikowski}. 
Considered as distributions in the difference variable this means that their wave front sets are 
\[\WF(W_{\mu})=\{(x,k)\in T^{*}\M|x\cdot x =0, \langle k,x\rangle=0, k\in\overline{V_+}\setminus 0\}\ , \]
where $\M$, at every point $x$, is identified with its tangent space and $\overline{V_+}$ denotes the closed forward lightcone in momentum space. 
The corresponding star products
\be\label{prod}
F\star_\mu G [\ph] \doteq \sum_{n=0}^\infty \frac{\hbar^n}{n!}\left<F^{(n)}[\ph],W_{\mu}^{\otimes n} G^{(n)}[\ph]\right>
\ee
are defined on a larger class of functionals, including in particular the local ones. An appropriate class is formed by the so-called microcausal functionals $F\in\cF_{\mu c}$. These are functionals $F$ where all functional derivatives $F^{(n)}$ exist as symmetric distributional densities
 whose wavefront sets satisfy the condition
 \be\label{eq:microcausal}
 \WF(F^{(n)})\cap \M^n\times \overline{V_+}^n=\emptyset\ .
 \ee
 On the space of regular functionals, all these products are equivalent. To see this, consider the linear invertible map 
 $\alpha_{\sst H_\mu}:\cF_{\rm reg}[[\hbar]]\rightarrow \cF_{\rm reg}[[\hbar]]$,
\[
\alpha_{H_\mu}\doteq e^{\frac{\hbar}{2}\cD_{H_\mu}}\,,
\]
where, in terms of formal integral kernels
\[
\cD_{H_\mu}\doteq \left\langle H_\mu,\frac{\delta^2}{\delta\ph^2}\right\rangle=\int H_\mu(x,y) \frac{\delta^2}{\delta \ph(x) \delta \ph(y)} \,.
\]
Then 
\[\alpha_{H_\mu}(F\star G)=\alpha_{\sst H_\mu} F\star_\mu \alpha_{\sst H_\mu} G\ .\]
Via the linear isomorphism 
\[
\beta_{\mu'/\mu}\doteq
\alpha_{H_{\mu^\prime}-H_\mu}
=e^{-\frac{\hbar}{4\pi}\int\ln(\frac{\mu'}{\mu}) \frac{\delta^2}{\delta\ph^2}}\ ,
\]
the star products $\star_\mu$ are mutually equivalent not only on the regular functionals, but on the larger space of microcausal functionals $\cF_{\mu c}[[\hbar]]$.
This is due to the fact that $H_{\mu'}-H_{\mu}$ is smooth. 

We can identify the elements of the algebra $(\cF_{\rm reg}[[\hbar]],\star)$
with normal ordered functionals
\be
\no{F}{}_{\mu}=(\alpha_\mu)^{-1}(F)\ .
\ee
with $F\in \cF_{\rm reg}[[\hbar]]$.
They satisfy the relation
\be\label{no-product}
\no{F}{}_\mu\star \no{G}{}_\mu=\no{F\star_\mu G}{}_{\mu}\ .
\ee
We now enlarge the algebra by elements $\no{F}_{\mu}$ with $F\in\cF_{\mu c}$ with the relations
\be \no{F}{}_\mu+\no{G}{}_\mu=\no{F+G}{}_\mu\ ,\ \lambda\no{F}{}_\mu=\no{\lambda F}{}_\mu\ ,\ \no{F}{}_{\mu}^*=\no{\overline{F}}{}_\mu\ee
and \eqref{no-product} for the product and obtain a *-algebra $\fA$.
This extension  may be understood as a completion in a suitable topology as discussed in \cite{BR16}. In general, the elements of $\fA$ can no longer be interpreted as functionals. 

In the following we set $\mu=1$. This means that the normal ordered elements are dimensionful objects which 
get numerical values only after the choice of a length unit. To simplify the notation we write $\no{F}{}_1=\no{F}$.

Now, given $g \in \cDd (\M)$, $a \in \R$, we consider the so-called vertex operators  $V_a(g)$ as normal ordered 
version of the functionals
\be\label{vertex}
v_{a}(g):\ph \mapsto \int e^{ia \ph(x)}\, g(x) \ .
\ee
i.e. $V_a(g)=\no{v_a(g)}$.
Note that, by normal ordering at $\mu=1$,  the vertex operators get mass dimension $\frac{\hbar a^2}{4\pi}$. 

The functional derivatives of the functionals of $v_a(g)$ are given by 
\[
\langle v_a(g)^{(n)}[\ph],h^{\otimes n}\rangle=i^na^n v_a(h^ng) \ ,
\] 
and satisfy the WF set condition \eqref{eq:microcausal} imposed for microcausal functionals. The star products of vertex operators converge if $\hbar |a|^2<4\pi$,
\begin{align}
&\left(v_{a_1}(g_1)\star_1 v_{a_2}(g_2)\right)[\ph]\\
&=\left(\sum_{n=0}^{\infty}\frac{\hbar^n}{n!}\left(\frac{a_1a_2}{4\pi}\right)^n\int(\ln(-(x-y)\cdot(x-y)+i\epsilon(x^0-y^0)))^ng_1(x)g_2(y)e^{i(a_1\ph(x)+a_2\ph(y))}\right)\\
&=\int (-(x-y)\cdot(x-y)+i\epsilon(x^0-y^0))^{\frac{\hbar a_1a_2}{4\pi}}g_1(x)g_2(y)e^{i(a_1\ph(x)+a_2\ph(y)}\ .
\end{align} 
The interaction Lagrangian of the sine-Gordon model is given in terms of vertex operators by
\[
V(g)\doteq \frac 1 2(V_a(g)+V_{-a}(g))\,.
\]
with $a>0$.

\subsection{States}
States are defined as normalized linear forms on the algebra. An interesting class of states are the {\em quasifree} ones. They are labeled by a symmetric real valued bisolution $H$ of the field equation which dominates $\Delta$ in the sense that
\be\label{quasifree}
\langle f,\Delta g\rangle^2\le 4\langle f,Hf\rangle\langle g,Hg\rangle\ee
for real valued test densities $f,g$. Here $H$ induces, as before $\Delta$, a linear map from test densities to smooth solutions. Given such $H$, every field configuration $\ph$ induces a state via the prescription 
\be\label{eq:omH}
\omega_{H,\ph}(\no{F})=\alpha_{H-H_1}(F)[\ph]
\ee

The bi-distribution $H_\mu$ given by \eqref{Hmu} satisfies the inequality \eqref{quasifree} and $\langle f, H_\mu f\rangle \geq 0$ only for test densities whose integral vanishes. Hence, in~\cite{Schubert}, Schubert modified $H_\mu$ by choosing 
a test density  $\psi\in\cDd(\M)$ with integral 1 and setting for  $r> 0$, 
\be \label{eq:HSchub}
H(x,y)=H_\mu(x,y)-H_\mu\psi(x)-H_\mu\psi(y)+\int\psi H_\mu\psi+\frac{1}{2r^2}\Delta\psi(x)\Delta\psi(y)+\frac{r^2}{2}\ .\ee,
$H$ is again a bi-solution, and since it differs from $H_\mu$ only by a smooth function, it satisfies the Hadamard condition. Contrary to $H_\mu$ it is positive semi-definite. Before proving that indeed, it satisfies \eqref{quasifree},  we introduce some notation. In the following, we use the Dirac bra-ket notation by writing
\be
|f\rangle\langle g|h \rangle\doteq f\langle g,h\rangle
\ee
whenever $g$ and $h$ are in duality to each other. We also introduce the projection $P_{\psi}:\cDd(\M)\rightarrow \cDd(\M)$,
\be\label{eq:DefPpsi}
P_{\psi}=1-|\psi\rangle\langle1|
\ee
and its transpose  (acting on $\cE(\M) \subset \cDd^\prime(\M)$)   $P_{\psi}^T:\cE(\M)\rightarrow\cE(\M)$, 
\be
P_{\psi}^T=1-|1\rangle\langle\psi|\ .
\ee
In this notation, we have
\be
H=P_{\psi}^TH_{\mu}P_{\psi}+\frac{1}{2r^2}|\Delta\psi\rangle\langle\Delta\psi|+\frac{r^2}{2}|1\rangle\langle 1|
\ee
and
\be
\Delta=P_{\psi}^T\Delta P_{\psi}+|\Delta\psi\rangle\langle 1|-|1\rangle\langle\Delta\psi|\ ,
\ee
both understood as maps $\cDd(\M) \rightarrow \cE(\M)$.

\begin{lemma}
The bi-distribution $H$ satisfies the condition \eqref{quasifree}.
\end{lemma}
\begin{proof}
It suffices to show that for real valued test densities $f$ and $g$ the matrix
\be
A=\left(
\begin{array}{cc}
\langle f,Hf\rangle & \langle f,\frac12\Delta g\rangle\\
\langle f,\frac12\Delta g\rangle & \langle g,Hg\rangle 
\end{array}\right)
\ee
is positive semidefinite.

From the fomulas above, we see that $A$ can be written as a sum of two matrices, where the first matrix  is obtained by replacing $H$ by 
$P_{\psi}^TH_{\mu}P_{\psi}$ and $\Delta$ by $P_{\psi}^T\Delta P_{\psi}$. This matrix is positive semidefinite, since $H_{\mu}$ satisfies the condition \eqref{quasifree} and $\langle f, H_\mu f\rangle \geq 0$ for test densities with vanishing integral. The second matrix has nonnegative trace and determinant and is therefore also positive semidefinite.  Hence $A$ as a sum of two positive semidefinite matrices is positive semidefinite.
\end{proof}

\subsection{Derezinski-Meissner-representation}

Inspired by a construction known from string theory, in~\cite{DM06} Derezinski and Meissner construct a representation of the free field on the Hilbert space
\be
\mathcal{H}=\mathcal{H}_0\otimes L^2(\mathbb{R})
\ee
where $\mathcal{H}_0$ is the usual Fock space for field derivatives and where $L^2(\mathbb{R})$ is a Hilbert space that describes the missing degree of freedom. On this Hilbert space, the field $\Phi$ is represented  by 
\be\label{eq:RepDM}
\Ppsi(f)\equiv\pi_{\psi}(\Phi(f))=\phic(P_{\psi}f)\otimes 1+1\otimes  q\langle1,f\rangle-1\otimes p\langle f,\Delta \psi\rangle\ .
\ee    
Here, $f$ is a test density, $\psi$ is a test density with total integral 1 as in Schubert's construction above, $P_\psi$ is the corresponding projection operator \eqref{eq:DefPpsi}, and $\phic$ is the canonical free massless field on Fock space, restricted to test densities with vanishing total integral, and  $q$ and $p$ are the standard position and momentum operators in the Schr\"odinger representation,
\begin{align*}
p\,\xi(k)&=k\xi(k)\,,\\
q\,\xi(k)&=i\hbar \frac{\partial}{\partial k}\xi(k)\,,
\end{align*}
where $\xi\in L^2(\R)$.

We briefly recall that \eqref{eq:RepDM} is indeed a representation of the free massless field. 
First of all, the field equation $\square\Ppsi(f)=\Ppsi(\square f)=0$ is satisfied,
since $\phic(\square f)=0$, $\text{\small $\int$} \square f=0$ and $\Delta\square f=0$. For the commutator we find
\begin{align*}
\frac{1}{i\hbar }[\Ppsi(f),\Ppsi(g)]=&\frac{1}{i\hbar }[\phic(P_{\psi}f),\phic(P_{\psi}g)]\otimes 1+\frac{1}{i\hbar}\otimes [q,p](-\langle1,f\rangle\langle g,\Delta \psi\rangle+\langle1,g\rangle\langle f, \Delta \psi\rangle)\\
=&\langle f,\Delta g\rangle-\langle1, f\rangle\langle\psi,\Delta g\rangle-\langle f,\Delta\psi\rangle\langle1,g\rangle
+\langle1,f\rangle\langle\psi,\Delta\psi\rangle\langle1,g\rangle\\
&-\langle1, f\rangle\langle g,\Delta \psi\rangle+\langle1,g\rangle\langle f,\Delta \psi\rangle\,,
\end{align*}
hence by the antisymmetry of $\Delta$ it follows that

\[
[\Ppsi(f),\Ppsi(g)]=i\hbar \langle f,\Delta g\rangle\ .
\]

Observe that  we have  changed the notation compared to~\cite{DM06}. Our $q$ corresponds to the $p$ of~\cite{DM06}, and our $p$ to their $-\chi$. Moreover, the functions $\sigma_{R,L}$ of DM are obtained from Schubert's $\psi$ in terms of their Fourier transforms by
\be\hat{\sigma}_R(k)=\int e^{i(k,k)x}\psi(x)\ ,\ \hat{\sigma}_L(k)=\int e^{i(k,-k)x}\psi(x)\ .\ee

To see the connection with Schubert's states we consider
\be\label{eq:defSchubertVectorState}
\Omega=\Omega_0\otimes \Omega_r\in\mathcal H\ee
where $\Omega_r$ is the ground state vector of the harmonic oscillator of mass $m$ and with frequency $\omega=\frac 1{mr^2}$ and $\Omega_0$ is the Fock vacuum. 
A short calculation shows that the resulting  2-point function is the one which Schubert constructed, 
\be\langle\Omega,\Ppsi(x)\Ppsi(y)\Omega\rangle=H(x,y)\ .\ee

To interpret the operators $p$ and $q$ in terms of the field, we first  observe that
\be\Ppsi(\psi)=1\otimes q\ .\ee
$p$ can be interpreted as the charge of the current $\partial_{\mu}\Ppsi$, which is conserved due to the field equation. More precisely, we prove the following lemma:

\begin{lemma}\label{lem:Charge}
Let $\chi^0,\chi^1$ be compactly supported test functions satisfying $\int\chi^0(t)dt=1$ and $\chi^1(\x)=1$ for all $\x$ with $|\x|<1$, and for $\lambda>0$, set $\chi_{\lambda}(t,\x) =-\lambda^2\dot{\chi^0}(\lambda t)\chi^1(\lambda^2 \x)\,dtd\x$.

Then	
\be 
Q_{\lambda}=\text{\small $\int$} \dot{\Phi}_\psi(t,\x)\lambda\chi^0(\lambda t)\chi^1(\lambda^2 \x)
\,dtd\x \equiv\Ppsi(\chi_{\lambda}) \, ,\ee
approximates the charge associated to $\partial_{\mu}\Ppsi$, and 
$1\otimes p$ is the total charge, in the sense that for all $\alpha\in\mathbb{R}$
	\be\mathrm{s}-\lim_{\lambda\to 0} e^{i\alpha Q_{\lambda}}=1\otimes e^{i\alpha p}\ .\ee

\end{lemma}
\begin{proof}

The idea to approximate charges in this manner goes back to Requardt \cite{Requardt}.

	Since $\int \chi_{\lambda}=0$ we have
	\be Q_{\lambda}=\phic(\chi_{\lambda})\otimes 1-1\otimes p \text{\small $\int$}  \chi_{\lambda}\left(\Delta\psi\right)\ .\ee
	First we show that for sufficiently small $\lambda>0$, $Q_{\lambda}=\phic(\chi_{\lambda})\otimes 1+1\otimes p$.
	Inserting the formulas for $\Delta$  and $\chi_{\lambda}$ into  the integral in the last term yields
	\be-\frac12\text{\small $\int$}  \lambda^2\dot{\chi^0}(\lambda t)\chi^1(\lambda^2 \x)\left(\Theta(t'-t-\x'+\x)-\Theta(t-t'-\x+\x')\right)\psi(t',\x')\, dtd\x dt'd\x'\ .\ee
	We perform the $t$-integral and obtain
	\be=-\frac12\text{\small $\int$} \lambda\left(\chi^0(\lambda(-t'-\x+\x'))+\chi^0(\lambda(t'+\x-\x'))\right)\chi^1(\lambda^2 \x)\psi(t',\x')\,  d\x dt'd \x'\ee
	and, after a suitable substitution,
	\be=-\frac12\text{\small $\int$}  \lambda\chi^0(\lambda s)\left(\chi^1(\lambda^2(-t'+\x'-s)+\chi^1(\lambda^2(\x'-t'+s))\right)\psi(t',\x')\,  dsdt'd\x'\ .\ee
	For $\lambda$ sufficiently small $\chi^1$ assumes the value 1 for all $t\in \lambda^{-1}\mathrm{supp}\chi^0$ and $(t',\x')\in\mathrm{supp}\psi$. Thus the integral is
	\be-\text{\small $\int$} ds\chi^0(s)\text{\small $\int$} dt'd\x'\psi(t',\x')=-1\ .\ee
	It remains to prove that $e^{i\alpha\phic(\chi_{\lambda})}$ converges strongly to $1$ for $\lambda\to 0$. This follows from the strong continuity of Weyl operators in the Fock representation, since $\phic(\chi_{\lambda})\Omega_0$, where $\Omega_0$ is the vacuum vector in the Fock space $\mathcal{H}_0$, converges to zero in the limit $\lambda\rightarrow 0$, as shown by the following computation where $\widehat{\phantom{f}}$ denotes the Fourier transform:
	\be\text{\small $\int$} \frac{dp}{|p|}|\widehat{\chi_{\lambda}}(|p|,p)|^2=\text{\small $\int$} \frac{dp}{|p|}|\lambda p\widehat{\chi^0}(\lambda^{-1}|p|)|^2|\frac{1}{\lambda^2}\widehat{\chi^1}(\frac{p}{\lambda^2})|^2=\lambda^2\text{\small $\int$} dp|p|\widehat{\chi^0}(\lambda |p|)|^2|\widehat{\chi^1}(p)|^2\to0\ee 
	in the limit $\lambda\to0$,  namely $\widehat{\chi^0}(\lambda |p|)$ tends to 1 for all $p$, and the integrand is bounded  by an integrable function which does not depend on $\lambda$.
\end{proof}

\begin{prop}The DM representation is irreducible. 
\end{prop}

\begin{proof}
We use the fact that $\mc B(\mc H)$ is the weak closure of the algebra generated by the Weyl operators, 
\be
W(f;\alpha,\beta):=e^{i\phic(f)}\otimes e^{i\alpha q+i\beta p}
\ee 
with test densities $f$ with $\int f=0$ and $\alpha,\beta\in\R$. This means that they form an irreducible set of operators on $\mathcal{H}$, i.e. they have trivial commutant in $\mc B(\mc H)$. Hence, the claim follows once we show that these operators are in the weak closure of the image of $\pi_{\psi}$. To see that this is true, choose a compactly supported test density  $\chi_\lambda$ as above in Lemma~\ref{lem:Charge} such that 
for sufficiently small $\lambda>0$,  we have $\int \chi_\lambda\Delta\psi=-1$. Then
\be
\pi_{\psi}(e^{i\Phi(f+\alpha\psi+(\beta+\int f\Delta\psi)\chi_{\lambda})})=
W(f+(\beta+\int f\Delta\psi)\chi_{\lambda};\alpha,\beta)\ ,
\ee
and by the strong continuity of Weyl operators in the Fock representation the statement follows. 
\end{proof}

It follows in particular that the states given by \eqref{eq:defSchubertVectorState} are pure.\medskip

Regarding the dependence of the representation on the choice of $\psi$, we reproduce the result of~\cite{DM06} in our notation.

\begin{lemma}
Let $\psi, \tilde{\psi}$ denote  test densities with total integral 1. Then
	\be\mathrm{Ad}(V_{\tilde{\psi},\psi})e^{i\Phi_{\tilde \psi}(f)}=e^{i\Ppsi(f)} \qquad \forall f\ee
	with
	\be V_{\tilde{\psi},\psi}=e^{i\phic(\xi) p-\frac{i}{2}p^2\text{\small $\int$}\tilde{\psi}\Delta\psi}\ ,\ \xi=\tilde{\psi}-\psi\ .\ee
\end{lemma}
\begin{proof}
To simplify notation, we use $p,q, \phic$ instead of $1 \otimes p$, $1 \otimes q$, $\phic\otimes 1$ for the operators on $\mc H$.
	We use the factorization
	\be e^{i\Phi_{\tilde \psi}(f)}=e^{i\phic(f-\tilde{\psi}\text{\small $\int$}f)} e^{i(q\text{\small $\int$}f-p\text{\small $\int$}f\Delta\tilde{\psi})}\ee
	and compute
	\be\mathrm{Ad}(e^{i(\phic(\xi) p)})\left(e^{i\phic(f-\tilde{\psi}\text{\small $\int$}f)}\right)=\left(e^{i\phic(f-\tilde{\psi}\text{\small $\int$}f)- ip\text{\small $\int$} \xi\Delta (f-\tilde{\psi}\text{\small $\int$f)}}\right)\ ,\ee
	\be\mathrm{Ad}(e^{i(\phic(\xi) p)})\left(e^{i(q\text{\small $\int$}f-p\text{\small $\int$}f\Delta\tilde{\psi})}\right)=e^{i( q\text{\small $\int$}f-p\text{\small $\int$}f\Delta\tilde{\psi})+i\phic(\xi)\text{\small $\int$}f}\ee
	and 
	\be\mathrm{Ad}(e^{i\lambda p^2})(e^{i(q\text{\small $\int$}f-p\text{\small $\int$}f\Delta\tilde{\psi}})=e^{i(q\text{\small $\int$}f-p\text{\small $\int$}f\Delta\tilde{\psi}-2p\lambda\text{\small $\int$}f)}\ee
	The proposition follows by combining these formulas.
\end{proof}

We now analyze the symmetries of the theory. Clearly, the net of local algebras transforms covariantly under the the conformal group
\be G=\mathrm{Diff}^+(\mathbb{R})\times\mathrm{Diff}^+(\mathbb{R})\ .\ee
We want to see whether the corresponding automorphisms can be unitarily implemented.

Let $\chi_\pm\in\mathrm{Diff}^+(\mathbb{R})$ and let $\chi$ be the corresponding conformal transformation,
\be\chi(x)=\chi(t,\x)=(\frac12(\chi_+(t+\x)+\chi_-(t-\x)),\frac12(\chi_+(t+\x)-\chi_-(t-\x)))\ .\ee
The corresponding automorphism $\alpha_\chi$ acts on the field $\phi$ by
$\alpha_{\chi}(\phi(x))=\phi(\chi(x))$. In the DM representation with respect to $\psi$ we find
\be\Ppsi(\chi(x))=\phic(\delta_{\chi(x)}-\psi)+q-p\Delta\psi(\chi(x))\ee
Let $G_0$ be the subgroup of compactly supported conformal transformations, and let $U$ be the projective representation of $G_0$ on $\mathcal{H}_0$ which implements the conformal transformations on the derivative of the field \cite{BMT}.
We find
\be\mathrm{Ad}(U(\chi))(\Ppsi(x))=\phic(\delta_{\chi(x)}-\chi_*\psi)+q-p\Delta\psi(x)\ee
with the push forward $\chi_*$ defined as the pull back of the inverse,
\be{\chi}_* =(\chi^{-1})^*\ .\ee
We have $\Delta\psi(x)=(\Delta{\chi}_*\psi)(\chi(x))$, hence
\be\alpha_{\chi}=\mathrm{Ad}(V_{{\chi}_*\psi,\psi}U(\chi))\ .\ee
We now check that $\chi\mapsto V_{\chi_*\psi,\psi}$ is a cocycle on the conformal group $G_0$ with respect to $U$.
Let $\alpha^0_\chi=\mathrm{Ad}(U(\chi))$ and $\chi_1,\chi_2\in G$. Then
\be V_{{\chi_1}_*\psi,\psi}\alpha^0_{\chi_1}(V_{{\chi_2}_*\psi,\psi})=V_{{\chi_1}_*\psi,\psi}V_{{{\chi_{1}}_*(\chi_2}_*\psi),{\chi_1}_*\psi}=V_{(\chi_1\chi_2)_*\psi,\psi}\ee
We conclude that $\chi\mapsto V_{{\chi}_*\psi,\psi}U(\chi)\equiv U_{\psi}(\chi)$ is a projective representation of $G_0$ with the same central charge as $U$. This representation, however, is not irreducible since $p$ is in the commutant. Actually, it is a direct integral over the spectrum of $p$ of irreducible representations labeled by the charge $Q=p$.

\subsection{Local normality}
We now show that the DM representation, considered as a representation of the exponentiated CCR algebra (the Weyl algebra) of time zero fields, is locally normal with respect to the representations induced by vacuum states for the massive situation. This means that the restrictions of these representaions to fields in a compact region are quasiequivalent. We recall that two representations $\pi$ and $\pi'$ of some C*-algebra $\mathfrak{A}$ are quasiequivalent if $\pi(A)\to\pi'(A)$, $A\in\mathfrak{A}$ extends to an isomorphism of the weak closures. 

While local normality holds for different nonzero masses as shown by Fr\"ohlich and Eckmann \cite{EF74}, this is not the case for the vacuum representation in the massless case. The local normality of the DM representation with respect to the massive ones now indicates that it is possible to construct the local observable algebras of massive models (such as presumably the sine-Gordon model) in this representation. In older approaches to the construction of the sine-Gordon model, one had to introduce a volume cutoff or an auxiliary mass term to avoid the infrared singularities of the vacuum representation of the massless field and therefore lost control over the local von Neumann algebras.

The DM representation is induced by the Schubert state, hence a state that is  quasifree. Since this is also true for the vacuum state in the massive case, we can use the (necessary and sufficient)  Araki-Yamagami criterion  for quasiequivalence of representations that are induced by quasifree states \cite{ArakiY}. 

For this purpose we consider the direct sum $L(\mathbb{R})=\mathcal{D}(\mathbb{R},\C)\oplus\mathcal{D}(\mathbb{R},\C)$ of two spaces of smooth compactly supported complex-valued test functions endowed with complex conjugation as the antilinear involution
and with the hermitian
form
\be\gamma(f,g)=\int (\overline{f_1}g_2-\overline{f_2}g_1)\, d\mathbf x \ .\ee 
The restriction of $\gamma$ to the real subspace is the symplectic form known from other formulations of the CCR algebra. According to \cite{ArakiY}, local quasiequivalence of the representations induced by the vacuum (of mass $m$) and the Schubert state is equivalent to 
\begin{enumerate}
\item the symmetrized scalar products $\langle , \rangle_{m,sym}$  and $\langle , \rangle_{s,sym}$, induced by the massive vacuum of mass $m>0$ and by the Schubert state, respectively, induce the same topology on $L(I)$, for any compact interval $I\subset \R$, and 
\item the square roots of the operators that define the respective 2-point functions in terms of e.g. $\langle , \rangle_{m,sym}$ (by the Riesz representation theorem) on $L(I)$ differ by a Hilbert Schmidt operator.
\end{enumerate}

We first calculate the scalar products. The detailed argument, together with all the relevant conventions, is presented in the appendix. 

The vacuum state of the massive theory induces the positive definite scalar product 
\[
\langle f,g\rangle_m=\tsf 1 2 \int \left(\omega^{-1}\overline{\hat{f_1}}\tilde{g_1}+\omega\overline{\hat{f_2}}\tilde{g_2}\right)\,  d\mathbf k +\tsf i 2 \gamma(f,g)\ ,
\]
where $\omega=\omega(\mathbf k)=\sqrt{\mathbf k^2 +m^2}$, the frequency on the positive mass shell.
The Schubert state 
for the massless scalar field induces the scalar product\beqa
\langle f,g\rangle_{s}&=&
\tsf{1}{2}\int \left(|\mathbf k|^{-\frac12}\overline{\widehat{P_{\psi}f_1}}-i|\mathbf k|^{\frac12}\overline{\hat{f_2}}
\right)\left(|\mathbf k|^{-\frac12}\widehat{P_{\psi}g_1}+i|\mathbf k|^{\frac12}\hat{g_2}\right)\, d\mathbf k \\
&& + \ 
\tsf 1 2\int (r \overline{f_1}-\tsf{i}{r} \psi  \overline{f_2})d\mathbf x\int  (rg_1+\tsf{i}{r}\psi g_2)d\mathbf y
\eeqa
with $P_{\psi}(h)=h -\psi \int h$ for $h\in \in\mathcal{D}(\mathbb{R})$, analogously to~\eqref{eq:DefPpsi}, with $\psi$ replaced by a real testfunction on $\mathbb{R}$ (not $\R^2$) with total integral 1. This is readily calculated in the usual way from  the time zero field and momentum  (see appendix).

Note that indeed, the integrand in $k$ in the formula for $\langle , \rangle_s$ is not singular in $\mathbf k= 0$ since
for any $h \in \mathcal D(\R)$, there is a constant $C  \geq 0$ which depends on $\psi$ and the support of $P_\psi h$, such that $|\widehat{P_{\psi}h}(\mathbf k)| \leq  C \ |\mathbf k| \  \|h\|_{L^1(\R)}$, 
\beqan\label{eq:k0estimatePpsi}
|\widehat{P_{\psi}h}(\mathbf k)| &=
&|\int  e^{i\mathbf k\mathbf x} \ P_\psi h(\mathbf x) 
\,d\mathbf x|
\ = \ |\int  (e^{i\mathbf k\mathbf x}-1) \ P_\psi h (\mathbf x)
d\mathbf x |  \nonumber \\
&\le&  |\mathbf k| \left(\sup_{\mathbf x\in \supp P_\psi h} |\mathbf x|\right) 
\int |P_\psi h (\mathbf x)|d\mathbf x 
\le  
  C \  |\mathbf k| \   \|h\|_{L^1(\R)}
\eeqan
where $C= \left(\sup_{\mathbf x\in \supp P_\psi h} |\mathbf x|\right)  (1+ \|\psi\|_{L^1(\R)})$. Observe that in the first line we subtracted  a term that is 0 and in last step we used the triangle inequality.

Like the massive one,  this scalar product is positive definite, since
$\langle f,f\rangle_{s}=0$ implies $f_2=0$, $P_\psi f_1=0$ i.e. $f_1=(\int f_1 d\mathbf x)\psi$, and $\int f_1 d \mathbf x=0$. 

The symmetrized scalar products then are
\[
\langle f,g\rangle_{j,sym}=\tsf12(\langle f,g\rangle_j+\langle \bar g,\bar f\rangle_j) 
\quad \mbox{ for } j=m,s
\] 
such that
\[
\langle f,g\rangle_j=\langle f,g\rangle_{j,sym}+\tsf i 2 \gamma(f,g) \quad \mbox{ for } j=m,s\ .
\]
Explicitly, we thus have
\be\label{eq:massivSym}
\langle f,f\rangle_{m,sym}=\tsf{1}{2}\int \left(\omega^{-1}|\hat{f_1}|^2+\omega|\hat{f_2}|^2\right)\ d \mathbf  k
\ee 
and 
\be\label{eq:SchubSym}
\langle f,f\rangle_{s,sym}=\tsf 1 2\int \left(|\mathbf k|^{-1}|\widehat{P_{\psi}f_1}|^2+|\mathbf k|\hat{f_2}|^2\right)\, d\mathbf k +\frac{r^2}2\,\big|\!\!\int  f_1 d\mathbf x\, \big |^2+\frac{1}{2r^2}\,\big|\!\!\int  \psi f_2d\mathbf  x\,\big|^2 \ .
\ee

{\lemma $\langle ,\rangle_{m,sym}$ and $\langle ,\rangle_{s,sym}$ induce the same topology on $L(I)$ for any compact interval $I \subset \R$.}

\begin{proof}
Without loss of generality, we assume $\mathrm{supp}\psi\subset I$ for the density defining the representation.

We write 
\[
\langle f,g\rangle_{m,sym}=\langle f_1,g_1\rangle_1+\langle f_2,g_2\rangle_2
\]
with
\[
\langle a,b\rangle_{1}=\tsf 1 2 \int  \overline{\hat{a}(\mathbf k)}\omega(\mathbf k)^{-1}\hat{b}(\mathbf k) \, d\mathbf k
\]
and 
\[
\langle a,b\rangle_{2}=\tsf 1 2
\int \overline{\hat{a}(\mathbf k)}\omega(\mathbf k)\hat{b}(\mathbf k)\, d\mathbf k \ .
\]
and denote the correponding norms $\| \ \|_1$ and $\| \ \|_2$.

Let  $\chi$ be a testfunction with $\chi(x)=1$ for $x\in I$.
Then the scalar product $\langle\cdot,\cdot\rangle_{s,sym}$ on $L(I)$ can be written in terms these products as
\be\langle f,g\rangle_{s,sym}=\langle f_1,Ag_1\rangle_{1}+\langle f_2,Bg_2\rangle_2\ee
with
\be \label{eq:defA}
 A=P_{\psi}^T M_{\frac{\omega}{|\mathbf k|}}P_{\psi}+r^2|M_\omega\chi\rangle\langle M_\omega \chi|\ee
and 
\be \label{eq:defB}
B=M_{\frac{|\mathbf k|}{\omega}}+\frac{1}{r^2}|M_{\frac 1 \omega}\psi\rangle\langle M_{\frac 1 \omega}\psi| \ee
where the bra-ket notation 
refers to the respective scalar products, and where $M_\omega$, $M_{\frac{\omega}{|\mathbf  k|}}$ etc.  stands for multiplication with $\omega$, ${\frac{\omega}{|\mathbf  k|}}$ etc. in momentum space, i.e. $\widehat{M_\omega \chi}(\mathbf k) = \omega \widehat{\chi}(\mathbf k)$.

We have to show that $A$ and $B$ are bounded and invertible.
\begin{description}
\item[$A$ is bounded]
$\chi$ is smooth, hence its Fourier transform is quickly decreasing and thus $||M_\omega \chi||^2_1=\int  \omega(\mathbf k)|\hat{\chi}(\mathbf k)|^2\,d \mathbf k<\infty$. Therefore, the rank one operator $|M_\omega \chi\rangle\langle M_\omega \chi|$ is bounded. Regarding the first summand in $A$, first observe that $P_{\psi}$ is bounded, and $\frac{\omega}{|\mathbf k|}$ is bounded for $|\mathbf k|\ge m$. For $|\mathbf k|<m$, observe that for any $h \in \mc D(I)$, we have the estimate \eqref{eq:k0estimatePpsi}, i.e. $|\widehat{P_{\psi}h}(\mathbf k)| \leq  C \ |\mathbf k| \  \|h\|_{L^1(\R)}$ for some $C\geq 0$. The claim follows when we have shown that 
the $L^1$-norm of $h$ is bounded by a multiple of $\|h\|_1$. To see this, we use that $\chi \equiv 1$ on $I$, and conclude\beqa
\int |h(\mathbf x)| d\mathbf x &=& \int |h(\mathbf x)| \chi(\mathbf x) d\mathbf x = \langle  \widehat{|h|} , \hat \chi \rangle_{L^2} = 
\langle  \sqrt{\omega}^{-1}{\widehat{|h|}}, \sqrt{\omega}^{-1}{\omega \hat \chi} \rangle_{L^2}\\
&=& 
2 \langle  |h| , M_\omega  \chi \rangle_{1} 
\leq 2 \|h\|_1  \| M_\omega \chi\|_1 
\eeqa
by Plancherel and by the Cauchy-Schwarz inequality.

\item[$A\ge c1$ for some $c>0$] 
Since $\omega(k)>|k|$ we have 
\[ A\ge P_{\psi}^TP_{\psi}+r^2|M_\omega \chi\rangle\langle M_\omega\chi|\ .\]
Furthermore,
\[||h||_1=||P_{\psi}h+2\langle M_\omega \chi,h\rangle_1\psi||_1\le||P_{\psi} h||_1+2|\langle M_\omega\chi,h\rangle_1|||\psi||_1\]
for any $h \in \mathcal D(I)$.
Hence
\[
||h||_1^2\le\mathrm{max}(2,\frac{2||\psi||_1^2}{r^2})(||P_{\psi}h||^2_1+2r^2|\langle M_\omega \chi,h\rangle_1|^2)\le\mathrm{const}\langle h,Ah\rangle_1\ .
\]

\item[$B$ is bounded]
This is obvious since $|\mathbf k|<\omega(\mathbf k)$ and $||M_{\frac 1 \omega} \psi||_1<\infty$. 
\item [$B\ge c1$ for some $c>0$] We show that on $L^2(I)$ the inequality
\[M_{|\mathbf k|}>c\, M_\omega\]
holds for some $c>0$. For this purpose we consider the operator 
\[ H=\mathbf x^2+ M_{|\mathbf k|}
\]
on $L^2(\mathbb{R})$. In momentum space (where $x=i\hbar \partial_k)$, this is the Schr\"odinger operator for a 1d particle in a potential $|\mathbf k|$. Its ground state energy $a$ is the first zero of the derivative of the Airy function, multiplied by $-1$.
Arguing as in \cite{BF02} we conclude that for $h\in L^2(\mathbb{R})$ with $||h||=1$ the inequality
\[
\langle h,M_{|\mathbf k|} h\rangle\ge \frac{\mathrm{const}}{\sqrt{\langle h,(\mathbf x-a)^2h\rangle}}
\]
hold. For $h\in L^2(I)$ we obtain the estimate
\[
\langle h,M_{|\mathbf k|}h\rangle\ge\frac{2\mathrm{const}}{|I|} \ .
\]
But then there exists another constant such that
\[
M_{|\mathbf k|}\ge \mathrm{const}M_{(|\mathbf k|+m)}\ge\mathrm{const}M_{\omega} \ .
\]
\end{description}
\end{proof}

{\lemma 
Denote by $C_m$ and $C_s$ the operators defining the 2-point functions given by a massive vacuum state and the Schubert state, respectively, with respect to the scalar product $\langle,\cdot,\cdot\rangle_{m,sym}$, i.e.
\[ \langle f,g\rangle_m=\langle f,C_m g\rangle_{m,sym}
\]
 and
 \[
\langle f,g\rangle_{s}=\langle f, C_s g\rangle_{m,sym}
\ .
\]
The square roots of these operators differ by a Hilbert Schmidt operator.}

\begin{proof} 
We first observe that 
\[
 C_{m}=
\left(\begin{array}{cc} 1& iM_\omega\\
                                       -iM_{\omega^{-1}} & 1
         \end{array}\right)  \] 
and
\[
 C_s=
 \left(\begin{array}{cc} A& iM_\omega\\
                                       -iM_{\omega^{-1}} & B
           \end{array}\right)    \ .
\]
with $A$ and $B$ the bounded operators \eqref{eq:defA} and \eqref{eq:defB} from the proof of the lemma above. By a lemma due to Buchholz \cite{Buchholz74},  the square roots of $C_m$ and $C_s$ differ by a Hilbert Schmidt operator if $C_m-C_S$ is of trace class, so we  are done when we show that $A-1$ and $B-1$ are trace class operators. 
 \begin{description}
 \item[$A-1$ is trace class]
Since the remaining term is a finite rank operator, it suffices to show that $M_{\sqrt{\omega|\mathbf k|^{-1}-1}}P_{\psi}$ is a Hilbert Schmidt operator on $\mathcal{H}_1(I)$, the completion of $\mathcal{D}(I)$ with respect to the scalar product $\langle\cdot,\cdot\rangle_{1}$ from 
the proof above. 

This is certainly the case if 
 \[ A'=M_{\sqrt{|\mathbf k|^{-1}-\omega^{-1}}}\,P_{\psi}
 \chi M_{\omega^{\frac12}}
\]
is Hilbert Schmidt on $L^2(\mathbb{R})$. Here, $\chi$ is understood  as a multiplication operator in position space, which turns into a convolution operator in Fourier space.  Hence, the integral kernel of $A'$ in Fourier space is
 \[ A'(\mathbf k,\mathbf p)=\sqrt{|\mathbf k|^{-1}-\omega^{-1}}(\hat{\chi}(\mathbf k-\mathbf p)-\hat{\psi}(\mathbf k)\hat{\chi}(\mathbf p))\omega^{\frac12}(\mathbf p)\ .\] 
The claim follows once we show that the function $(\mathbf k,\mathbf p)\mapsto A'(\mathbf k,\mathbf p)$ is square integrable. For $|\mathbf k|>m$ this is obvious, hence we restrict the function to $|\mathbf k|\le m$.
 We split the difference in the middle of the term above in the form
 \[
 \hat{\chi}(\mathbf k-\mathbf p)-\hat{\psi}(\mathbf k)\hat{\chi}(\mathbf p)=(1-\hat{\psi}(\mathbf k))\hat{\chi}(\mathbf k-\mathbf p)+\hat{\psi}(\mathbf k)(\hat{\chi}(\mathbf k-\mathbf p)-\hat{\chi}(\mathbf p))
 \]
and obtain a decomposition of $A'$ into two terms.
 Both terms are square integrable. For the first term we exploit $\hat{\psi}(0)=1$, hence $|1-\hat{\psi}(\mathbf k)|\le\mathrm{const}|\mathbf k|$, and for the second we  use the mean value theorem and the fact that $\hat \chi$ and all its derivatives are  rapidly decreasing, such that
\[
|\hat{\chi}(\mathbf p+\mathbf k)-\hat{\chi}(\mathbf p)|\le|\mathbf k|\ \mathrm{sup}_{0\le\lambda\le 1}|\frac{d}{d\mathbf p}\hat{\chi}(\mathbf p+\lambda \mathbf k)|\le|\mathbf k|\frac{c}{(||\mathbf p|-m|+1)^N}
\]
for $|\mathbf k|\leq m$.
 \item[$B-1$ is trace class]
 Proceeding as in the previous case consider the integral kernel
 \[
 B'(\mathbf k,\mathbf p)=\sqrt{\omega(\mathbf k)-|\mathbf k|}\hat{\chi}(\mathbf k-\mathbf p)\omega(\mathbf p)^{-\frac12}\ .
 \]
 We set $\mathbf q=\mathbf k-\mathbf p$ and have to show that
 the map 
 \[
(\mathbf k,\mathbf q)\to \sqrt{\omega(\mathbf k)-|\mathbf k|}\hat{\chi}(\mathbf q)\omega(\mathbf k-\mathbf q)^{-\frac12}\] 
 is square integrable. We consider first the $\mathbf k$ integral. We have 
 \[
\omega(\mathbf k)-|\mathbf k|=\frac{m^2}{\omega(\mathbf k)+|\mathbf k|}\ ,
\]
therefore the $\mathbf k$ integral is polynomial bounded in $\mathbf q$,
\[
\int \frac{\omega(\mathbf k)-|\mathbf k|}{\omega(\mathbf k-\mathbf q)}\, d\mathbf k \le \mathrm{const} \ln{(1+|\mathbf q|)} \ .\]
Since $\hat{\chi}$ decreases quickly, the square integrability follows.
\end{description} 
\end{proof}

\section{Interacting local net of the sine-Gordon model}
\subsection{Formal $S$-matrix in the DM representation}

We will now see how in the DM representation, using a certain class of states, we can further improve the estimates from~\cite{BR16} on the $S$-matrix and the interacting fields of the sine-Gordon model.

\subsubsection{Convergence}  
Our starting point is the abstract algebra $\fA$ generated by vertex operators, as defined in \cite{BR16}. These operators are the ones occuring in the $S$-matrix of the sine-Gordon model, which as was already given in~\cite{BR16} is 
\beqan\label{eq:TnHSGV}
S(\lambda V(g))&=&\sum_n \tsf 1 {n!} 
(\tfrac{i\lambda}{\hbar})^n (\tsf 1 2)^n\ \cT_n((V_{a}(g) + V_{-a}(g))^{\otimes n}) 
\nonumber
\\ &=&
\sum_n \lambda^n\underbrace{\tsf 1 {n!}  (\tfrac{i}{\hbar})^n (\tsf 1 2)^n \sum_{k=0}^n \left(n \atop k\right) \cT_n\Big(V_{a}(g)^{\otimes k}\otimes V_{-a}(g)^{\otimes  (n-k)}\Big)}_{\displaystyle \doteq  S_n(V(g))}
\,,
\eeqan
for $g \in \cDd(\M)$. Here, $\cT_n$ abbreviates the $n$-th order time ordered product, which (following \cite{BR16} and using the notation introduced in section \ref{Algebra}) is given by
\[
\cT_n(V_{a_1}(g) \otimes \dots \otimes V_{a_n}(g))=\no{\cT^{H_1}_n(v_{a_1}\otimes\dots \otimes v_{a_n})(g^{\otimes n})}\,,
\]
where
\begin{multline}
\cT_n^{H_\mu} \left(v_{a_1} \otimes \dots \otimes v_{a_n}\right)(g^{\otimes n})\\= \int e^{i(a_1\Phi(x_1) + \dots + a_n\Phi(x_n))}
\underbrace{e^{- \sum_{1\leq i< j\leq n}a_i a_j \hbar \Delta^{\rm F}_\mu(x_i,x_j)}}_{\displaystyle\doteq t^\mu_{n,\underline{a}}(\underline{x})}
g^{\otimes n}(\underline{x}) 
\,,\label{eq:TnHexplicit}
\end{multline}
and
 \beq\label{eq:Feynm}
\Delta^{\rm F}_\mu(x,y) = - \tsf 1 {4 \pi} \ln(-\mu^2(x-y)^2 - i\epsilon)
\eeq
is the Feynman propagator. We obtain
\be\label{TnH}
\cT_n(V_{a_1} \otimes \dots \otimes V_{a_n})(g^{\otimes n})=
\int \no{e^{i(a_1\Phi(x_1) + \dots + a_n\Phi(x_n))}}
\ t^1_{n,\underline{a}}(\underline{x}) \ 
g^{\otimes n}(\underline{x}) 
\,.
\ee

Before proving convergence, we explain how to extend the DM representation to operator-valued distributions of the form
\[
\no{e^{i(a_1\Phi(x_1) + \dots + a_n\Phi(x_n))}}{}_\mu
\]
To that end, let $h$ be a test density. We have
\[\no{e^{i\Phi(h)}}{}_\mu=e^{i\Phi(h)}e^{\frac12\langle h,H_\mu h\rangle}\]
and therefore 
\begin{align*}
|\langle \Omega_0\otimes \xi,\pi_\psi(\no{e^{i\Phi(h)}}{}_\mu)\Omega_0\otimes\xi\rangle|
&\le
e^{\frac12\langle h,H_\mu h\rangle}\langle\Omega_0,e^{i\phic(P_{\psi} h)}\Omega_0\rangle|\langle\xi,e^{iq\text{\small $\int$}\! h-ip\text{\small $\int$}\! h\Delta\psi}\xi\rangle| \nonumber \\
&\le e^{-\frac12 \langle\psi,H_\mu\psi\rangle(\text{\small $\int$}\!h)^2+\langle h,H_{\mu}\psi\rangle\text{\small $\int$}\!h}\int dp|\overline{\xi(p)}\xi(p-\!\text{\tiny $\int$}\! h)|
\end{align*}
where $\Omega_0$ is the Fock vacuum and $ \xi \in \mathrm{L}^2(\R)$.
Given a sequence of $(h_m)_m$ with $\int h_m=0$ for all $m$ that approximates the measure $\sum_{i=1}^n a_i (\delta_{x_i}-\delta_{y_i})$ in the sense of distributions
with fixed compact support, the last integral can be uniformly estimated by  $||\xi||^2$. We thus obtain the estimate
\be \label{eq:vertexDM}
|\langle \Omega_0\otimes \xi,\pi_{\psi}(\no{e^{i\sum_ja_j(\Phi(x_j)-\Phi(y_j))}})\Omega_0\otimes \xi\rangle|\le||\xi||^2\ .\ee
\medskip

\begin{prop}
Let $\psi$ be a test density with total integral 1 and consider the representation $\pi_\psi$ of the field on $\mathcal{H}=\mathcal{H}_0\otimes L^2(\mathbb{R})$ from~\cite{DM06}, which we recounted above in equation \eqref{eq:RepDM}. 
Let $ f\in \cD(\M)$, $\Omega_0$  the Fock vacuum and $ \xi \in \mathrm{L}^2(\R)$ with $||\xi||=1$.
Then
\[||\pi_{\psi}(S_n(V(g)))\Psifxi|| \le C^n (n!)^{\frac{1-p}{p}}\]
where $1<p<\frac{4\pi}{\hbar a^2}$ and $C>0$ depends on $g$, $a$ and $p$.
\end{prop}
\begin{proof}
We use the expansion above of $S_n(V(g))$ into a sum of time ordered products of vertex operators and compute their absolute squares
\[\cT_n(V_{a_1} \otimes \dots \otimes V_{a_n} )(g^{\otimes n})^*\star\cT_n(V_{a_1} \otimes \dots \otimes V_{a_n} )(g^{\otimes n})\]
\[=\int  g^{\otimes n}(\underline{x})g^{\otimes n}(\underline{y})\no{e^{i\sum_j(\Phi(x_j-\Phi(y_j))}}\prod_{i<j}|(x_i-x_j)^2|^{\frac{\hbar a_ia_j}{4\pi}}|(y_i-y_j)^2|^{\frac{\hbar a_ia_j}{4\pi}}
\prod_{i,j=1}^n|(x_i-y_j)^2|^{-\frac{\hbar a_ia_j}{4\pi}}\]
\[\times e^{i(\sum_{i<j}a_ia_j(\Delta_D(x_i-x_j)-\Delta_D(y_i-y_j))-\frac12\sum_{i,j=1}^n\Delta(x_i-y_j))}\]
with the Dirac propagator 
\[\Delta_D=\frac12(\Delta_R+\Delta_A)\ .\]
Their contribution to the norm can be estimated by
\[
\|\pi_{\psi}(\cT_n(V_{a_1} \otimes \dots \otimes V_{a_n} )(g^{\otimes n}))\pi_{\psi}(e^{i\Phi(f)})\Omega_0\otimes \xi\|^2\]
\[\le \int  |g^{\otimes n}(\underline{x})g^{\otimes n}(\underline{y})\prod_{i<j}|(x_i-x_j)^2|^{\frac{\hbar a_ia_j}{4\pi}}|(y_i-y_j)^2|^{\frac{\hbar a_ia_j}{4\pi}}
\prod_{i,j=1}^n|(x_i-y_j)^2|^{-\frac{\hbar a_ia_j}{4\pi}}
\]
\[
\times |\langle \Omega_0\otimes \xi,\pi_{\psi}(\no{e^{i\sum_j(\Phi(x_j)-\Phi(y_j))}})\Omega_0\otimes \xi\rangle|
\]
since both $\Delta$ and $\Delta_D$ are real valued.

To estimate the expectation value in the last line, we use \eqref{eq:vertexDM}, and considering $\xi$ with  $||\xi||=1$, we obtain the bound
\[
\|\pi_{\psi}(\cT_n(V_{a_1} \otimes \dots \otimes V_{a_n} )(g^{\otimes n}))
\pi_{\psi}(e^{i\Phi(f)})\Omega_0\otimes \xi\|^2
\]
\[
\le \int  |g^{\otimes n}(\underline{x})g^{\otimes n}(\underline{y})\prod_{i<j}|(x_i-x_j)^2|^{\frac{\hbar a_ia_j}{4\pi}}|(y_i-y_j)^2|^{\frac{\hbar a_ia_j}{4\pi}}
\prod_{i,j=1}^n|(x_i-y_j)^2|^{-\frac{\hbar a_ia_j}{4\pi}}\ .
\]
In the expansion of $S_n(V(g))$ into time ordered products of smeared vertex operators only coefficients $a_i=\pm a$ occur. We may rename the coordinates such that $x_i$ and $y_i$ are exchanged if $a_i=-a$. We observe that the estimate is independent of the signs of the coefficients $a_i$ and find
\[||\pi_{\psi}(S_n(V(g)))\pi_{\psi}(e^{i\Phi(f)})\Omega_0\otimes \xi||^2\]
\[\le \left(\frac{1}{n!}\right)^2\int  |g^{\otimes 2n}(\underline{x},\underline{y})\prod_{i<j}|(x_i-x_j)^2|^{\frac{\hbar a^2}{4\pi}}|(y_i-y_j)^2|^{\frac{\hbar a^2}{4\pi}}
\prod_{i,j=1}^n|(x_i-y_j)^2|^{-\frac{\hbar a^2}{4\pi}}\ .\]
We can now use the results of \cite{BR16} (where an older estimate of Fr\"ohlich \cite{Frohlich} in the Euclidean theory was adapted to the Minkowski signature) and find
\[||\pi_{\psi}(S_n(V(g)))\pi_{\psi}(e^{i\Phi(f)})\Omega_0\otimes \xi|| \le C^n (n!)^{\frac{1-p}{p}}\]
where $1<p<\frac{4\pi}{\hbar a^2}$ and $C>0$ depends on $g$, $a$ and $p$. 
\end{proof}
As a consequence of this estimate, the expansion of the S-matrix $S(\lambda V(g))$ converges in the representation $\pi_{\psi}$ strongly on the dense domain $D$ which is spanned by vectors of the form $\Psifxi$ and defines  an operator $S(\lambda V(g))_{\psi}$. 
\begin{prop}\label{prop:SpsiUnitary}
$S(\lambda V(g))_{\psi}$ is unitary.
\end{prop}
\begin{proof}
The formal power series $S(\lambda V(g))$ is unitary. This means that
\[\sum_{n+m=k}S_n^*S_m=\delta_{k0}=\sum_{n+m=k}S_nS_m^*\]
with 
\[S_n=S_n(V(g))\ .\]
We have for $\Psi\in D$
\[||S(\lambda V(g))_{\psi}\Psi||^2=\sum_{n,m}\lambda^{n+m}\langle \pi_{\psi}(S_n(V(g))\Psi,\pi_{\psi}(S_m(V(g))\Psi\rangle\]
\[=\sum_k\lambda^k\langle\Psi,\pi_{\psi}(\sum_{n+m=k}S_n^*S_m)\Psi\rangle=||\Psi||^2\ ,\]
hence $S(\lambda V(g))_{\psi}$ is isometric and has a unique extension to an isometry on $\cH$.
Analoguous arguments can be used for the adjoint power series and yield that also the adjoint of $S(\lambda V(g))_\psi$ is an isometry. 
This proves the claim.
\end{proof}

\subsection{Construction of local observables}
In this section we construct local algebras of observables for the sine-Gordon model. 

We follow the prescription given in \cite{Vienna} on how to construct the interacting local net of observables, given a family of unitaries that are interpreted as local S-matrices, which we recall here.

Let $\sD\doteq\Gamma_c(E\rightarrow \M)$ be the space of test objects over $\M$ (compactly supported sections of some vector bundle $E$ over $\M$). Consider unitaries {$S(f)$}, $f\in\sD$
with $S(0)=0$, which generate a *-subalgebra  of {$\fA$} and satisfy for $f,g,h\in\sD$  Bogoliubov's factorization relation
\begin{equation}\label{Bogoliubov}
{S(f+g+h)=S(f+g)S(g)^{-1}S(g+h)} 
\end{equation}
if the past {$J_-$} of {$\supp h$} does not intersect  {$\supp f$}
(or, equivalently, if the future {$J_+$} of {$\supp f$} does not intersect  {$\supp h$}).
\begin{defi}
	Define the relative S-matrices as
	\begin{equation}\label{relS}
	{f\mapsto S_g(f)=S(g)^{-1}S(g+f)}\ .
	\end{equation}
\end{defi}
\begin{defi}[\cite{Vienna}]
	The Haag-Kastler net {$\fA_g$} of the interacting theory is then defined by
	the local algebras {$\fA_g(\cO)$} which are generated by the relative S-matrices $S_g(f),\supp f\subset \cO$.
\end{defi}
Note that $g$ plays a role of cutoff function that labels local interactions and $S_g(f)$ is interpreted as the retarded observable under the influence of the interaction labeled by $g$. Next we take the algebraic adiabatic limit.
\begin{defi}[\cite{Vienna}]
	 Let $G\in\Gamma(E\rightarrow \M)$ (no support restriction). Set
	\begin{equation}\nonumber
	[G]_\cO=\{g\in\sD^n|g\equiv G \text{ on a neighborhood of }J_+(\cO)\cap J_-(\cO)\} \ .
	\end{equation}
	We consider the $\fA$-valued maps
	\begin{equation}\nonumber
	S_{G,{\cO}}(f):[G]_{\cO}\ni g\mapsto S_g(f)\in\fA\ .
	\end{equation}
	The local algebra {$\fA_{G}(\cO)$} is defined to be the algebra generated by {$S_{G,{\cO}}(f),\supp f\subset \cO$}. 
\end{defi}
The interpretation as ``adiabatic limit'' follows from the fact that $G$ can be set to be constant and this  corresponds to removing the cutoff from the interaction.
\begin{thm}[\cite{Vienna}]
The net  {$\fA_{G}(\cO)$} with $G=\const$ satisfies the Haag-Kastler axioms Isotony, Covariance and Locality, i.e.
\begin{description}
\item[Isotony] For each inclusion $\cO_1\subset\cO_2$ there exists an injective homomorphism
\[i_{\cO_2\cO_1}:\fA_{G}(\cO_1)\to\fA_{G}(\cO_2)\]
such that $i_{\cO_3\cO_2}\circ i_{\cO_2\cO_1}=i_{\cO_3\cO_1}$.
\item[Covariance]
For each Poincar\'{e} transformation $L$ there exist isomorphisms
\[\alpha_L^{\cO}:\fA_G(\cO)\to\fA_G(L\cO)\]
such that
\[\alpha_L^{\cO_2}\circ i_{\cO_2\cO_1}=i_{L\cO_2 L\cO_1}\circ\alpha_L^{\cO_1}\ .\]
\item[Locality]
If $\cO_1,\cO_2$ are spacelike separated subsets of $\cO$, then
\[[i_{\cO\cO_1}(\fA_G(\cO_1)),i_{\cO\cO_2}(\fA_G(\cO_2))]=\{0\}\ .\]
\item[Time-Slice-Axiom]
Let $\cO_1\subset\cO_2$ be globally hyperbolic regions such that $\cO_1$ contains a Cauchy surface of $\cO_2$. Then the homomorphism $i_{\cO_2\cO_1}$ is an isomorphism. 
\end{description}
\end{thm}
\subsubsection{Local $S$-matrices of the sine-Gordon model}
We start with specifying the label set of test objects. We concentrate on three classes of interacting fields: the scalar field $\Phi$ itself, the interaction Lagrangian $\cos{a\Phi} $ and the term $\sin{a\Phi}$ occuring in the field equation. This amounts to consider test objects
\[
(g,h)\in\sD\doteq\cD(\M,\C)\oplus\cD(\M,\R)\ .
\]
and to define a map $L:\sD\rightarrow \cF_{\loc}$ by
\[
L(g,h)\doteq v_a(g)+v_{-a}(\overline{g})+\Phi(h)\ .
\]
In particular, for $g$ real valued, $\no{L(g,h)}$ is the interaction term of the sine-Gordon model. Using the functional formalism, we define the S-matrices as
\[
\cS(g,h)\equiv \mathcal{T} e^{i\no{L(g,h)}}\doteq \sum_{n=0}^{\infty}\tfrac{1}{n!}\left(\tfrac{i}{\hbar}\right)^n\cT_n(\no{L(g,h)}\nolimits^{\otimes n})=\sum_{n=0}^{\infty}\tfrac{1}{n!}\left(\tfrac{i}{\hbar}\right)^n\no{\cT^{H_1}_n(L(g,h)^{\otimes n})}\,.
\]
In Proposition~\ref{prop:SpsiUnitary} we have already shown that $\cS(g,h)_{\psi}$ is a well defined unitary operator for $g$ real valued and $h=0$, so it remains to prove that the same holds for arbitrary $\cS(g,h)_{\psi}$, $(g,h)\in\sD$. The estimate for complex valued $g$ and $h=0$ is identical. To include the general case, we use the fact that on the level of formal power series we have
\[\mathcal{T} e^{i\no{L(g,h)}}=\cS(g,0)\T \mathcal{T} e^{i\Phi(h)}\ .\]
where
\[\mathcal{T}e^{i\Phi(h}=e^{i\Phi(h)}e^{\frac{i}{2}\langle h,\Delta_D h\rangle}\]
\begin{prop}
For all $(g,h)\in\sD$, $\cS(g,h)_\psi$ is a well defined unitary operator on the Hilbert space $\cH$.
\end{prop}
\begin{proof}
We have
\[\no{e^{i\sum_i a_i\Phi(x_i)}}\T\  e^{i\Phi(h)}=\no{e^{i\sum_i a_i\Phi(x_i)}}\star\  e^{i\Phi(h)}e^{i\sum_i a_i\Delta_D h(x_i)}\ .\]	
Inserting this formula into the estimate for
\[||\pi_{\psi}(\cT_n(\no{L(g,0)}\nolimits^{\otimes n})\T\  e^{i\Phi(h)})\pi_{\psi}(e^{i\Phi(f)})\Omega_0\otimes \xi||^2\]
we can repeat the arguments which show convergence of the sum of norms and also the unitarity of the sum.
\end{proof}

We have shown the existence of $S(f)$, $f\in\mathscr{D}$; for the construction of the interacting net we still need to show that they satisfy the factorization relation. We use the following general fact. 
\begin{prop}
Let $D\subset\tilde{D}\subset\cH$ be dense subspaces. Consider bounded operators  $O$ on $\cH$ which are defined by series of endomorphisms $O_n$ of $\tilde{D}$ which converge together with their adjoints strongly on $D$, i.e. the sequences 
$O\Psi=\sum_{n=0}^{\infty} O_n\Psi$ and  $O^*\Psi=\sum_{n=0}^{\infty} O^*_n\Psi$ both converge in norm for all $\Psi\in D$.
Let $O,O',O''$ be three such operators with the property that
\[\sum_{n+m=k}O_nO'_m=O''_k\ .\]
Then $OO'=O''$.
\end{prop}
\begin{proof}
Let $\Psi,\Psi'\in D$. Then
\[\langle\Psi',OO'\Psi\rangle=\langle O^*\Psi,O'\Psi\rangle=\sum_{n,m} \langle O^*_n\Psi', O'_m\Psi\rangle=\sum_{n,m}\langle\Psi',O_nO'_m\Psi\rangle\]
\[=\sum_k\langle \Psi',\sum_{n+m=k}O_nO'_m\Psi\rangle=\sum_k\langle\Psi', O''_k\Psi\rangle=\langle\Psi',O''\Psi\rangle\]
Since $D$ is dense, the proposition follows. 
\end{proof}
To prove the factorization relation it now suffices to show that also the $\star$-products of S-matrices and their adjoints converge in the representation $\pi_{\psi}$. This amounts to the same estimates as above. By the proposition it then follows that also the relations between the corresponding unitary operators hold.

We finally define a net of von Neumann algebras associated to the sine-Gordon model. For a fixed bounded region $\cO$ we choose some $g\in G_{\cO}$ with $G=\const$ and consider the seminorms
\[||A||_{\Psi,g}=|\langle\Psi,A(g)\Psi\rangle|\]
The set of seminorms does not depend on the choice of $g$. We therefore can complete each algebra 
$\fA_G(\cO)$ and obtain a net of von Neumann algebras with normal embeddings. 

\section{Relation to the Thirring model}
The (massless) Thirring model is a theory of a massless Dirac field in 2 dimensions with a current-current interaction. It is closely related to the massless scalar field. Nevertheless, its history is quite involved, with a lot of partially contradicting treatments. Especially fascinating is that this relation extends to a corresponding relation between the massless sine-Gordon model and the massive Thirring model, first described by Coleman
\cite{Coleman75}. Usually one even claims equivalence between these theories, but this remains vague in the absence of a precise definition of equivalence (see \cite{Benfatto} for the state of the art).

One problem treated in the literature which induced a lot of confusion is the absence of a vacuum state for the massless scalar field. If the vacuum is replaced by the  Poincar\'{e} invariant linear, but non-positive functional induced by the 2-point function $H_{\mu}$, one has difficulties to prove the positivity
of the Wightman functions of the Thirring model, as pointed out by Wightman \cite{Wightman67} and finally solved by Carey et al \cite{Carey85}.

In our treatment we start from the realization of the massless scalar field in terms of functionals on the space of smooth functions on Minkowski space.
We need in addition a dual field $\tilde{\Phi}$ which 
satisfies the relation $\partial_{\mu}{\Phi}=-\epsilon_{\mu\nu}\partial^{\nu}\tilde{\Phi}$ 
with the antisymmetric symbol $\epsilon_{\mu\nu}$ with $\epsilon_{01}=1$. Instead of imposing conditions on the field configurations which guarantee the existence of $\tilde{\phi}$,
we double the configuration space and consider functionals of pairs of smooth functions $(\phi,\tilde{\phi})$ which are a priori independent from each other. Later we divide out the subspace $\mathcal{I}$ of functionals vanishing on pairs which are solutions of the wave equation and satisfy the condition
\be
\partial_{\mu}{\phi}=-\epsilon_{\mu\nu}\partial^{\nu}\tilde{\phi} \ .
\ee
In terms of the lightcone variables $u=x^0+x^1$ and $v=x^0-x^1$ this means 
\be\label{eq:defphitildeLC}
\partial_u\tilde{\phi}=\partial_u\phi\ ,\ \partial_v\tilde{\phi}=-\partial_v\phi\ .
\ee
We then introduce a $\star$-product which extends the Weyl-Moyal $\star$-product of the theory of a single massless field such that $\mathcal{I}$ becomes an ideal.
We set
\[
\tilde{\Phi}(x)\star\tilde{\Phi}(y)-\tilde{\Phi}(x)\tilde{\Phi}(y)=\frac{i}{2}\Delta(x,y)
\]
and
\[
\tilde{\Phi}(x)\star\Phi(y)=\tilde{\Phi}(x)\Phi(y)+\frac{i}{2}\tilde{\Delta}(x,y)\ ,
\]
\[
\Phi(x)\star\tilde{\Phi}(y)=\Phi(x)\tilde{\Phi}(y)-\frac{i}{2}\tilde{\Delta}(y,x)\ ,
\]
with 
\be
\tilde{\Delta}(x,y)=\frac{1}{2}(\Theta(-u)-\Theta(v)) 
\ .
\ee
(where here and in the following $(x-y)=\frac12(u+v,u-v)$).
We see that $\tilde{\Phi}$ is not relatively local to $\Phi$. One could modify $\tilde{\Delta}$ by adding a constant. The chosen version later turns out to be convenient.

It is now easy to find functionals with fermionic commutation relations. 
For this purpose we pass from the Weyl-Moyal $\star$-product $\star$ to the Wick $\star$-product $\star_{\mu}$
which is  induced by the linear isomorphism $e^{\Gamma_{\mu}}$,
\[
F\star_{\mu}G=e^{\Gamma_{\mu}}(e^{-\Gamma\mu}F\star\ e^{-\Gamma_\mu}G)\ ,
\]
with 
\[
\Gamma_\mu=\int\frac12
\left(
\begin{array}{cc}
	\frac{\delta}{\delta\phi}&\frac{\delta}{\delta\tilde{\phi}}
\end{array}
\right)
\left(
\begin{array}{cc}
	H_{\mu}&\tilde{H}\\
	\tilde{H}&H_{\mu}
\end{array}
\right)
\left(
\begin{array}{c}
	\frac{\delta}{\delta\phi}\\
	\frac{\delta}{\delta\tilde{\phi}}
\end{array}
\right) \ . 
\]
Here 
\be
\tilde{H}(x,y)=-\frac{1}{4\pi}\ln\left|\frac{u}{v}\right|\ .
\ee
The operator $e^{-\Gamma_\mu}$ performs the normal ordering of functionals of the field,
\be
\no{F}\nolimits_{\mu}=e^{-\Gamma_\mu}(F)
\ee
for regular functionals $F$. 
As in Section 2, we set $\mu=1$ and obtain dimensionful normal ordered functionals.
The Wick $\star$-product can be extended to all local functionals of the fields $\phi$ and $\tilde{\phi}$. 
Note that $\tilde{H}$ is not Lorentz invariant. Therefore normal ordering is changed by Lorentz transformations, and we obtain the adapted action of Lorentz transformations
\be
(\sigma_{\Lambda(\theta)}\no{F})(\phi,\tilde{\phi})=e^{-\frac{\theta}{2\pi}\frac{\delta^2}{\delta\phi\delta\tilde{\phi}}}\no{F}(\phi\circ\Lambda(\theta),\tilde{\phi}\circ\Lambda(\theta))
\ee
with
\[
\Lambda(\theta)=
\left(\begin{array}{cc}
	\cosh\theta & \sinh\theta\\
	\sinh\theta & \cosh \theta
\end{array}\right)\ .
\]
Exponential functions of linear combinations of the scalar fields (''vertex operators'') 
transform as
\[
\sigma_{\Lambda(\theta)}\no{e^{i(\alpha\Phi(x)+\beta\tilde{\Phi}(x))}}=e^{\frac{\alpha\beta}{2\pi}\theta}\no{e^{i(\alpha\Phi(\Lambda(\theta)x)+\beta\tilde{\Phi}(\Lambda(\theta)x))}}\ .
\]
For their $\star$-product we obtain
\begin{multline}\label{eq:prod_exp}
\no{e^{i(\alpha\Phi(x)+\beta\tilde{\Phi}(x))}}\star \no{e^{i(\alpha'\Phi(y)+\beta'\tilde{\Phi}(y))}}=\no{e^{i(\alpha\Phi(x)+\beta\tilde{\Phi}(x)+\alpha'\Phi(y)+\beta'\tilde{\Phi}(y))}} \\
\times(i u+\epsilon)^{\frac{(\alpha+\beta)(\alpha'+\beta')}{4\pi}}(i v+\epsilon)^{\frac{(\alpha-\beta)(\alpha'-\beta')}{4\pi}}e^{\frac{i(\alpha\beta'-\alpha'\beta)}{4}}\ .
\end{multline}
and thus
\begin{align}\label{eq:anti}
\no{e^{i(\alpha\Phi(x)+\beta\tilde{\Phi}(x))}}\star\no{e^{i(\alpha'\Phi(y)+\beta'\tilde{\Phi}(y))}}&=\no{e^{i(\alpha\Phi(y)+\beta\tilde{\Phi}(y))}}\star\no{e^{i(\alpha'\Phi(x)+\beta'\tilde{\Phi}(x))}}\nonumber\\
&\times e^{-i(\alpha\alpha'+\beta\beta')\Delta(x,y)-i(\alpha\beta'+\alpha'\beta)\tilde{\Delta}(x,y)+i\alpha\beta'}\ .
\end{align}
Since, for spacelike separated arguments, $\Delta$ vanishes and $\tilde{\Delta}$ assumes the values $0$ and $1$, we see that the vertex operators anticommute if $\alpha\beta',\alpha'\beta=\pm\pi$.

We now set for given $\alpha>0$
\beqan
\psi_{+}(x)&=&-i(2\pi)^{-\frac12}\no{e^{i(\alpha\Phi(x)+\frac{\pi}{\alpha}\tilde{\Phi}(x))}} 
\\
\psi_{-}(x)&=&(2\pi)^{-\frac12}\no{e^{i(-\alpha\Phi(x)+\frac{\pi}{\alpha}\tilde{\Phi}(x))}} 
\eeqan
and consider them as dimensionful fields, as in Section 2.  The prefactor $-i$ for $\psi_+$ will turn out later to be convenient for a simple choice of $\gamma$-matrices (see \cite{Mandelstam}).

These fields and their adjoints anticommute according to \eqref{eq:anti}.  
Moreover, they transform under Lorentz transformation as
\[
\sigma_{\Lambda(\theta)}\psi_{\pm}(x)=e^{\pm\frac{\theta}{2}}\psi_{\pm}(\Lambda(\theta)x)\ .
\]

\paragraph{\textbf{Specialization to $\mathbf{\alpha=\sqrt{\pi}}$}}. We now show that the case $\alpha=\sqrt{\pi}$  corresponds to free chiral massless Fermi fields. The anticommutation relations are
\beqa
\{\psi_+(x),\psi_+(y)\}&=\ 0 \ =& \{\psi_-(x),\psi_-(y)\}\ ,
\\
\{\psi_+(x),\psi_-(y)\}&=\ 0\ =& \{\psi_+^*(x),\psi^*_-(y)\}\ ,
\eeqa
and
\be
\{\psi_+^*(x),\psi_+(y)\}=e^{i\sqrt{\pi}(-\Phi(x)+\Phi(y)-\tilde{\Phi}(x)+\tilde{\Phi}(y))}\delta(u)=\delta(u)
\ee
since by \eqref{eq:defphitildeLC}, $\Phi+\tilde{\Phi}$ does not depend on $v$, and analogously
\be
\{\psi_-^*(x),\psi_-(y)\}=\delta(v)\ .
\ee
Moreover, again by \eqref{eq:defphitildeLC}, we have for $\alpha=\sqrt{\pi}$,
\beqa
\partial_v\psi_+(x)&=&i(\partial_v\Phi(x)+\partial_v\tilde{\Phi}(x))\psi_+(x)=0 \\
\partial_u\psi_-(x)&=&i(-\partial_u\Phi(x)+\partial_u\tilde{\Phi}(x))\psi_+(x)=0
\eeqa

From the operator product expansions of these Fermi fields, we can reconstruct the derivatives of $\Phi$,
\be
\lim_{y\to x}\psi_+^*(x)\star\psi_+(y)-(2\pi)^{-1}(iu+\epsilon)^{-1}=-\pi^{-\frac12}\partial_u\Phi(x)
\ee 
and
\be
\lim_{y\to x}(\psi_-^*(x)\star\psi_-(y)-(2\pi)^{-1}(iu+\epsilon)^{-1}=\pi^{-\frac12}\partial_v\Phi(x)
\ee 
and likewise, exponential functions of $\Phi$ by  
\be
\psi_-^*(x)\star\psi_+(x)=\frac{1}{2\pi}e^{i2\sqrt{\pi}\Phi(x)}
\ee
and
\be
\psi_+^*(x)\star\psi_-(x)=\frac{1}{2\pi}e^{-i2\sqrt{\pi}\Phi(x)} \ .
\ee

\paragraph{\textbf{Generic $\mathbf{\alpha>0}$.}} 
From \eqref{eq:prod_exp} we obtain
\begin{align*}
&\psi_+^*(x)\star\psi_+(y)\\
=&(2\pi)^{-1}
(iu+\epsilon)^{-\frac{(\alpha+\frac{\pi}{\alpha})^2}{4\pi}}(iv+\epsilon)^{-\frac{(\alpha-\frac{\pi}{\alpha})^2}{4\pi}}e^{-i(\alpha(\Phi(x)-\Phi(y))+\frac{\pi}{\alpha}(\tilde{\Phi}(x)-\tilde{\Phi}(y))}\\
=&(2\pi)^{-1}(iu+\epsilon)^{-1}
(-uv+i(u+v)\epsilon)^{-\frac{(\alpha-\frac{\pi}{\alpha})^2}{4\pi}}e^{-i(\alpha(\Phi(x)-\Phi(y))+\frac{\pi}{\alpha}(\tilde{\Phi}(x)-\tilde{\Phi}(y))}\ .
\end{align*}
As in the previous case of free Fermions we expand the exponential function for $x\approx y$ and find
for spacelike separated points an expansion of the form
\[
\psi_+^*(x)\star\psi_+(y)=(2\pi)^{-1}|uv|^{-(\frac{(\alpha-\frac{\pi}{\alpha})2}{4\pi}}\left((iu+\epsilon)^{-1}+(\alpha+\frac{\pi}{\alpha})\partial_u\Phi(x)+\ldots\right)
\]
where the omitted terms vanish in the spacelike coincidence limit $y\stackrel{K}{\to} x$ (i.e. there is a closed spacelike cone $K$ such that $y-x\in K$).
We obtain the formula
\[
N(\psi_+^*\psi_+)(x):=\lim_{y\stackrel{K}{\to} x}|(y-x)^2|^{\frac{(\alpha-\frac{\pi}{\alpha})2}{4\pi}}\left(\psi_+(x)\star\psi_+(y)-(2\pi)^{-1}(iu)^{-1}\right)=(\frac{\alpha}{2\pi}+\frac{1}{2\alpha})\partial_u\Phi(x) \ .
\]

By an analogous argument we obtain
\[
N(\psi_-^*\psi_-)(x)=-(\frac{\alpha}{2\pi}+\frac{1}{2\alpha})\partial_v\Phi(x)\ .
\]
For the mixed products we find
\[
\lim_{y\to x}|(x-y)^2|^{\frac{\alpha^2}{4\pi}-\frac{\pi}{4\alpha^2}}\psi_+^*(x)\star\psi_-(y)=\frac{1}{2\pi}e^{-2i\alpha\Phi(x)}
\]
and
\[
\lim_{y\to x}|(x-y)^2|^{\frac{\alpha^2}{4\pi}-\frac{\pi}{4\alpha^2}}\psi_-^*(x)\star\psi_+(y)=\frac{1}{2\pi}e^{2i\alpha\Phi(x)}\ .
\]

\paragraph{\textbf{Equations of the motion}} For the equations of motion, we find
\be\label{eq:eom1}
\partial_v\psi_+(x)=i(\alpha-\frac{\pi}{\alpha})(\partial_v{\Phi}(x))\psi_+(x)
\ee
and
\be\label{eq:eom2}
\partial_u\psi_-(x)=-i(\alpha-\frac{\pi}{\alpha})(\partial_u{\Phi}(x))\psi_-(x)\ .
\ee
We want to interpret these equations as the field equations of the Thirring model. We consider 
\be
\psi=\left(\begin{array}{c}
	\psi_+\\
	\psi_-
\end{array}\right)
\ee
as a Dirac field.
As discussed before, Lorentz transformations act on these fields as
\be
\Lambda(\theta)\cdot\psi(x)=\left(\begin{array}{c}
	e^{\frac{\theta}{2}}\psi_+(\Lambda(\theta)x)\\
	e^{-\frac{\theta}{2}}\psi_-(\Lambda(\theta)x)
\end{array}\right)
\ee 
The conjugate Dirac field can be defined by
\be
\overline{\psi}=(\psi_+^*,\psi_-^*)\gamma^0=(\psi_-^*,\psi_+^*)
\ee
with the matrix
\be
\gamma^0=\left(\begin{array}{cc}
	0&1\\
	1&0\end{array}\right)\ .
\ee
For $\gamma^1$ we choose
\be
\gamma^1=\left(\begin{array}{cc}
	0&-1\\
	1&0\end{array}\right)\ .
\ee
Then the 2-vector with entries $\gamma^0$ and $\gamma^1$ transforms under Lorentz transformations as a point in 2d-Minkowski space.

The field equation in the classical Thirring model is
\be
\gamma^{\mu}(i\partial_\mu-g(\overline{\psi}\gamma_{\mu}\psi))\psi=0\ .
\ee
Inserting the above definitions we find the coupled system of equations
\begin{align}
i\partial_u\psi_-=&g\psi_+^*\psi_+\psi_-\\
i\partial_v\psi_+=&g\psi_-^*\psi_-\psi_+\ .
\end{align} 
For the quantum theory we have to replace the classical currents $j_{u,v}=\psi_{\pm}^*\psi_{\pm}$ by the suitably normalized normal products $N(\psi_{\pm}^*\psi_{\pm})$ such that the charge associated to the fields $\psi_{\pm}$ is equal to $-1$. Then 
\be
j_{u}=\frac{\alpha}{\pi}\partial_{u}{\Phi}\ , \ j_{v}=-\frac{\alpha}{\pi}\partial_{v}{\Phi}
\ee
hence the equations of motion (\ref{eq:eom1},\ref{eq:eom2}) coincide with the equations of motion for the Thirring model with the coupling constant 
\be
g=\frac{\pi^2}{\alpha^2}-\pi .
\ee
The interaction term of the sine-Gordon model coincides with a mass term in the Thirring model,
\be
\no{\cos{\beta\Phi(x)}}=\pi N(\overline{\psi}\psi)(x)
\ee 
with 
\be
N(\overline{\psi}\psi)(x)=\lim_{y\to x}|(x-y)^2)|^{\frac{\alpha^2}{4\pi}-\frac{\pi}{4\alpha^2}}\psi_+^*(x)\star\psi_-(y)
\ee
and $\alpha=\beta/2$. This observation is the basis for Coleman's argument for the equivalence of both models.

{\rem 
In our framework we see that the construction of the observable algebras of the sine-Gordon model also yields  the observables of the massive Thirring model. The local algebras of the massive Thirring model, however, are proper subalgebras of the algebra of the sine-Gordon model and consist only of elements which are invariant under the automorphism induced by $\phi\to\phi+2\pi/\beta$.
}

\section{Conclusion and Outlook}

We constructed the net of local observable algebras (Haag-Kastler net) for the 2 dimensional sine-Gordon model in the ultraviolet finite regime. In spite of many previous works on this model (see \cite{Benfatto} for an overview), this is, to the best of our knowledge, the first complete construction. It was obtained within the formalism of perturbative Algebraic Quantum Field Theory. The von Neumann algebras associated to bounded regions are subalgebras of the algebra of the free massless scalar field in the Derezinski-Meissner \cite{DM06} representation. 
In this representation there is neither a vacuum state for the free field (because of infrared problems) nor a vacuum state for the sine-Gordon model (due to Haag's Theorem). We proved that locally (i.e. restricted to local subalgebras) the representation (considered as a representation of the Weyl algebra of time-zero fields) is quasiequivalent to the vacuum representations of massive scalar fields. It is therefore to be expected that this remains true for the vacuum representation of the sine-Gordon model. We also showed that the formalism of pAQFT allows a treatment of the relation to the massive Thirring model where the formulae known from previous work \cite{Mandelstam} get a mathematically precize meaning.

Our work opens the perspective to investigate this model in more detail. We expect that the integrable structure of the classical model shows up in the existence of infinitely many conserved currents, as suggested from perturbation theory \cite{Lowenstein}. We already showed that the DM representation has a 1-parameter family of superselection sectors, and a countable subset should  represent the charged sectors of the Thirring model.

The major open problem is the existence of a vacuum representation. The problem to overcome is the slow decay of correlations in a framework starting from the massless field. In spirit, this is similar to a problem also pointed out by Hollands and Wald~\cite{HollandsWald}, who suggested that the
perturbation series for the operator product expansion (OPE) coefficients of a QFT might converge,
but that it is less clear how to construct states perturbatively. A good understanding of this problem is needed in our approach to the sine-Gordon model, and also other similar special 2-dimensional models such as the Gross-Neveu model or  $P(\phi)_2$-theory should be investigated in this spirit.

Once the vacuum representation is constructed, one could try to prove the factorization condition for the S-matrix in the sine-Gordon model. It would be important to relate our construction to the form factor program \cite{Smirnov} and to the Lechner program \cite{Lechner} for the construction of the model. We hope to come back to these problems in future work.


\subsection*{Acknowledgement}

DB and KR would like to thank the Isaac Newton Institute, Cambridge, (programme OAS: \textit{Operator algebras: subfactors and their applications}) and the MFO, Oberwolfach, (workshops 1737 and 1748) for kind hospitality. We would like to thank Daniela Cadamuro for interesting comments. This research was partially supported by KR's EPSRC grant \verb|EP/P021204/1| .


\medskip

\begin{appendix}

\section{Explicit formulae}\label{app:scalarP}

\paragraph{\textbf{Scalar products and time-zero fields}}
In this section we recall some basic facts about the Fock representation of the massive scalar field and compare these with analogous structures in the DM representation.

Let us start with fixing some notation. Let $f\in\cS(\R,\C)$. Define the smeared creation and annihilation operators by
\[
a(f)\doteq \int a(\mathbf{k}) \overline{f(\mathbf{k})} d\mathbf{k}\,,\quad a^\dagger(f)\doteq \int a^\dagger(\mathbf{k}) \overline{f(\mathbf{k})} \,.
\]
with the commutation relations
\[
[a(f),a^\dagger(g)]=\int \overline{f(\mathbf{k})} g(\mathbf{k}) d\mathbf{k}\,.
\]
Introduce a measure $\omega_m:\R\rightarrow \R_+$, $\omega_m(\mathbf{k})\doteq \sqrt{\mathbf \mathbf{k}^2+m^2}$. 

In the massive case, for $f_1,f_2\in\cD(\R,\C)$ we define 
\begin{align*}
\phi_m(f_1)\doteq \frac{1}{\sqrt{2}}(a^\dagger(\omega_m^{-1/2}\hat{f}_1)+a(\omega_m^{-1/2}\hat{\bar f}_1))\,,\\
{\pi}_m(f_2)\doteq \frac{1}{\sqrt{2}}(a^\dagger(i\omega_m^{-1/2}\hat{f}_2)+a(i\omega_m^{-1/2}\hat{\bar f}_2))\,,
\end{align*}
where we use the following definitions of the Fourier transform and its inverse:
\begin{align*}
\hat{f}(\mathbf{k})&\doteq \frac{1}{\sqrt{2\pi}} \int e^{-i\mathbf{k} \mathbf{x}} f(\mathbf{x}) d\mathbf{x}\,,\\
\check{f}(\mathbf{x})&\doteq \frac{1}{\sqrt{2\pi}} \int e^{i\mathbf{k} \mathbf{x}} g(\mathbf{k}) d\mathbf{k}\,.
\end{align*}
Hence, explicitly, the time-zero fields are represented by the following operator-valued distributions:
\begin{align*}
\phi_m(\mathbf{x})&=\frac{1}{\sqrt{2}}\int \frac{1}{\sqrt{2\omega_m(\mathbf{k})}}\left(e^{-i\mathbf{k}\mathbf{x}}a^\dagger(\mathbf{k})+e^{i\mathbf{k}\mathbf{x}}a(\mathbf{k})\right)d\mathbf{k}\\
\pi_m(\mathbf{x})&=\frac{1}{\sqrt{2}}\int \sqrt{\frac{\omega_m(\mathbf{k})}{2}}\left(e^{-i\mathbf{k}\mathbf{x}}a^\dagger(\mathbf{k})-e^{i\mathbf{k}\mathbf{x}}a(\mathbf{k})\right) d\mathbf{k}
\end{align*}
For the time-zero field and time-zero momentum in the DM representation we apply the  representation $\pi_\psi$  given by formula \eqref{eq:RepDM} to $\Phi(f)$ where $f=\delta_{t=0} f_1({\bf x})dtd\bf x$ and $f=\delta^\prime_{t=0} f_1({\bf x})dtd\bf x$ and where we choose a test density $\psi$ of the form $\delta_{t=0} \psi_1 ({\bf x})dt d\bf x$ and by abuse of notation denote $\psi_1$ again by $\psi$. This gives
\beqa
\phi_s(f_1)=\phic\left(\delta_{t=0} P_\psi f_1\right)\otimes 1  + 1\otimes q\langle 1 , f_1\rangle=\phi_0(P_\psi f_1) + 1\otimes q\langle 1
\eeqa
for the time-zero field, and 
\beqa
\pi_s(f_1)=\dot\phic\left(\delta_{t=0}  f_1\right)  
 - 1\otimes p \langle f_1, \psi\rangle= \pi_0(f_1) - 1\otimes p \langle f_1, \psi\rangle
\eeqa
for the time-zero momentum.

Now let $f=(f_1,f_2)\in L(\R)=\mathcal{D}(\mathbb{R}, \C)\oplus\mathcal{D}(\mathbb{R},\C)$. Define a family of operators
\[
B_j(f)\doteq \phi_j(f_1)+\pi_j(f_2)\,,
\]
where $j=m,s$.

We obtain the products $\left<.,.\right>_j$, $j=m,s$ by using the prescription:
\[
\left<f,g\right>_j\doteq \omega_j(B_j(\bar f)B_j(g))\,.
\]
This is a straightforward calculation, but for the convenience  of the reader we spell out some details in the case $j=s$. Recall that
\[
\omega_s(B_j(\bar f)B_j(g))=\left<\Omega, B_j(\bar f)B_j(g)\Omega\right>\,.
\]
where $\Omega$ is given by \eqref{eq:defSchubertVectorState}. Hence
\begin{multline*}
\omega_s(B_j(\bar f)B_j(g))=\left<\Omega,(\phi_s(\bar{f}_1)+\pi_s(\bar{f}_2))(\phi_s(g_1)+\pi_s(g_2))\Omega\right>\\
=\frac{1}{2}\int\overline{\left(\tfrac{1}{\sqrt{\omega_0}}\widehat{P_\psi f_1}+i\sqrt{\omega_0} \widehat{f}_2\right)}\left(\tfrac{1}{\sqrt{\omega_0}}\widehat{P_\psi g_1}+i\sqrt{\omega_0} \widehat{g}_2\right)\\+\left<\Omega_r,q^2\Omega_r\right>\int \bar{f}_1\int g_1+\left<\Omega_r,p^2\Omega_r\right>\int \bar{f}_2\psi\int g_2\psi\\
+\left<\Omega_r,\left(qp \int\bar{f}_1\int g_2\psi+pq\int \bar{f}_2\psi\int g_1\right)\Omega_r\right>
\end{multline*}
Note that $\omega_0=|\mathbf{k}|$. Inserting expectation values of $q^2$, $p^2$, $qp$, and rearranging, we obtain
\beqa
\langle f,g\rangle_{s}&=&
\tsf{1}{2}\int \left(|\mathbf k|^{-\frac12}\overline{\widehat{P_{\psi}f_1}}-i|\mathbf k|^{\frac12}\overline{\hat{f_2}}
\right)\left(|\mathbf k|^{-\frac12}\widehat{P_{\psi}g_1}+i|\mathbf k|^{\frac12}\hat{g_2}\right)\, d\mathbf k \\
&& + \ 
\tsf 1 2\int (r \overline{f_1}-\tsf{i}{r} \psi  \overline{f_2})d\mathbf x\int  (rg_1+\tsf{i}{r}\psi g_2)d\mathbf y\,.
\eeqa
\end{appendix}

\bigskip\noindent
{\sc \small Dorothea Bahns, Mathematisches Institut, Georg-August-Universit{\"a}t G{\"o}ttingen, Germany. \verb+dbahns@mathematik.uni-goettingen.de+
}

\medskip\noindent
{\sc \small Klaus Fredenhagen,
II. Institut f\"ur Theoretische Physik, Universit\"at Hamburg, Germany. 
\verb+klaus.fredenhagen@desy.de+
}
 
\medskip\noindent
{\sc \small Kasia Rejzner,
 	Department of Mathematics, University of York, United Kingdom. 
\verb+kasia.rejzner@york.ac.uk+
}


\begin{thebibliography}{99}
\bibitem{ArakiY}
H.~Araki and S.~Yamagami, ``On quasi-equivalence of quasifree states of the canonical commutation relations.'' Publications of the Research Institute for Mathematical Sciences {\bf 18} (1982) 703-758.
\bibitem{BR16}
D.~Bahns, K.~Rejzner,
``The Quantum Sine Gordon model in perturbative AQFT,''
Commun. Math. Phys. (2017). \url{https://doi.org/10.1007/s00220-017-2944-4}
\bibitem{Benfatto}
G.~Benfatto, P.~Falco, V.~Mastropietro,
``Massless Sine-Gordon and Massive Thirring Models:
Proof of Coleman's Equivalence''
Commun. Math. Phys. {\bf 285} (2009) 713--762 
\bibitem{BF02}
  R.~Brunetti and K.~Fredenhagen,
  ``Remarks on time energy uncertainty relations,''
  Rev.\ Math.\ Phys.\  {\bf 14} (2002) 897
\bibitem{Brunetti}
  R.~Brunetti, D.~Guido and R.~Longo,
  ``Modular localization and Wigner particles,''
  Rev.\ Math.\ Phys.\  {\bf 14} (2002) 759-786
\bibitem{Buchholz74}
  D.~Buchholz,
  ``Product States For Local Algebras,''
  Commun.\ Math.\ Phys.\  {\bf 36} (1974) 287.
\bibitem{BMT}
  D.~Buchholz, G.~Mack and I.~Todorov,
  ``The Current Algebra on the Circle as a Germ of Local Field Theories,''
  Nucl.\ Phys.\ Proc.\ Suppl.\  {\bf 5B} (1988) 20-56.
\bibitem{Cadamuro}
  D.~Cadamuro and Y.~Tanimoto,
  ``Wedge-local observables in the deformed Sine-Gordon model,''
  arXiv:1612.02073 [math-ph].
\bibitem{Carey85}  
A. ~L.~Carey, S.~N.~M.~Ruijsenaars, and J.~D.~Wright,
``The Massless Thirring Model: Positivity of Klaiber's
n-Point Functions,'' 
Commun.\ Math.\ Phys. {\bf 99} (1985) 347-364 
\bibitem{Coleman75}
S.~Coleman, 1992)
``Quantum Sine-Gordon equation as the massive Thirring model.'' 
Phys. Rev. {\bf D 11} (1975) 2088--2097.
\bibitem{DM06}
  J.~Derezinski and K.~A.~Meissner,
  ``Quantum massless field in 1+1 dimensions,''
  Lect.\ Notes Phys.\  {\bf 690} (2006) 107
  [math-ph/0408057].
\bibitem{EF74}
  J.~P.~Eckmann and J.~Fr\"ohlich,
  ``Unitary equivalence of local algebras in the quasifree representation,''
  Ann.\ Inst.\ H.\ Poincare Phys.\ Theor.\  {\bf 20} (1974) 201.
\bibitem{Faddeev}
L.~D~. Faddeev, V.~E.~Korepin, 
``Quantization of solitons,'' 
Theoretical and Mathematical Physics,\ {\bf 25} (1975) 103--1049
\bibitem{Vienna}
K.~Fredenhagen and K.~Rejzner,
``Perturbative construction of models of quantum field theory,''
in ``Advances in Algebraic Quantum Field Theory'' , R.~Brunetti et al (eds),
Mathematical Physics Studies, Springer 2015
\bibitem{Frohlich-Seiler}
  J.~Fr\"ohlich and E.~Seiler,
  ``The Massive Thirring-Schwinger Model (QED in Two-Dimensions): Convergence of Perturbation Theory and Particle Structure,''
  Helv.\ Phys.\ Acta {\bf 49} (1976) 889.
\bibitem{Frohlich} 
J.~Fr\"ohlich, 
``Classical and quantum statistical mechanics in one and two dimensions:
Two-component Yukawa  and Coulomb systems,''
Commun.\ Math.\ Phys.\ {\bf 47} (1976),  233-268.
\bibitem{Hadjiivanov}
  L.~K.~Hadjiivanov and D.~T.~Stoyanov,
  ``Wightman Functions In The Thirring Model,''
  Theor.\ Math.\ Phys.\  {\bf 46} (1981) 236-242.
\bibitem{HollandsWald} S.~Hollands, R.~Wald, ''Axiomatic quantum field theory in curved spacetime,'' 
  Commun.\ Math.\ Phys.\  {\bf 293} (2010) 85
\bibitem{Karowski}
M.~Karowski, P.~Weisz,
``Exact form factors in (1 + 1)-dimensional field theoretic models with soliton behaviour,''
Nucl.\ Phys.\ {\bf B 139} (1978) 455-476
\bibitem{Lechner}
  G.~Lechner,
  ``Construction of Quantum Field Theories with Factorizing S-Matrices,''
  Commun.\ Math.\ Phys.\  {\bf 277} (2008) 821-860
\bibitem{Lowenstein}
J.~H.~Lowenstein and E.~R.~Speer,
``Existence of Conserved Currents
in the Perturbative Sine-Gordon
and Massive Thirring Models,''
Commun.~Math.~Phys.\ {\bf 63} (1978) 97-112
\bibitem{Mandelstam}
 S.~Mandelstam,``Soliton operators for the quantized Sine-Gordon equation,''
 Physical\ Review\  {\bf D11}(1975) 3026-3030.
  \bibitem{Requardt}
  M.~Requardt,
  ``Symmetry Conservation and Integrals Over Local Charge Densities in Quantum Field Theory,''
  Commun.\ Math.\ Phys.\  {\bf 50} (1976) 259.
\bibitem{Radzikowski}
  M.~J.~Radzikowski,
  ``Micro-local approach to the Hadamard condition in quantum field theory on curved space-time,''
  Commun.\ Math.\ Phys.\  {\bf 179} (1996) 529-553.
\bibitem{Rejzner}
K.~Rejzner,
  ``Perturbative Algebraic Quantum Field Theory : An Introduction for Mathematicians,''
  Mathematical Physics Studies
New York: Springer (2016)

 \bibitem{Schroer}
  B.~Schroer,
  ``Modular wedge localization and the d = (1+1) form-factor program,''
  Annals Phys.\  {\bf 275} (1999) 190-223
 \bibitem{Schubert}
S.~Schubert,``\"Uber die Charakterisierung von
Zust\"anden hinsichtlich der
Erwartungswerte quadratischer
Operatoren,'' Diplomarbeit Hamburg 2013, available at \url{http://www.desy.de/uni-th/theses/Dipl\_Schubert.pdf}
\bibitem{Smirnov}
F.~A.~Smirnov, 
``Form Factors in Completely Integrable Models of Quantum Field Theory,''
World Scientific, Advanced Series in Mathematical Physics, Vol 14 (1992)
\bibitem{Wightman67}
A.~S.~Wightman, 
``Introduction to some aspects of the relativistic dynamics of quantized fields.''
In:
High energy electromagnetic interactions and field theory, pp. 171-289 Ltvy, M. ed. New York:
Gordon and Breach, 1967


\end{thebibliography}
\end{document}